\documentclass{article}

\usepackage[preprint]{neurips_2026}
\usepackage{amsmath}
\usepackage{amssymb}
\usepackage{mathtools}
\usepackage{amsthm}
\usepackage{dsfont}
\usepackage{todonotes}
\usepackage{csquotes}
\usepackage{thm-restate}
\usepackage{hyperref}
\usepackage{cleveref}
\usepackage{comment}
\usepackage{xspace}
\usepackage{paralist}

\AddToHook{cmd/appendix/before}{\crefalias{section}{appendix}}
\AddToHook{cmd/appendix/before}{\crefalias{subsection}{appendix}}

\newtheorem{theorem}{Theorem}[section]
\newtheorem{lemma}[theorem]{Lemma}
\newtheorem{corollary}[theorem]{Corollary}
\newtheorem{claim}[theorem]{Claim}
\newtheorem{observation}[theorem]{Observation}

\theoremstyle{definition}

\newtheorem{definition}[theorem]{Definition}

\newtheorem{conjecture}[theorem]{Conjecture}

\DeclareMathOperator*{\polylog}{poly\,log}

\DeclareMathOperator{\symDiff}{\Delta}
\newcommand{\R}{\ensuremath{\mathbb{R}}}
\newcommand{\N}{\mathds{N}}
\newcommand{\I}{\mathcal{I}}
\newcommand{\eps}{\ensuremath{\varepsilon}}
\newcommand{\our}{\texttt{FFTHungarian}\xspace}
\newcommand{\their}{\texttt{PAWL}\xspace}

\newcommand{\leftChild}[1]{\ensuremath{\mathsf{left}_{#1}}}
\newcommand{\rightChild}[1]{\ensuremath{\mathsf{right}_{#1}}}

\newcommand{\ko}[1]{\textcolor{red}{(KO: #1)}}
\newcommand{\mcs}[1]{\textcolor{teal}{(MCs: #1)}}
\newcommand{\sa}[1]{\textcolor{cyan}{(SA: #1)}}
\newcommand{\jc}[1]{\textcolor{blue}{(JC: #1)}}
\newcommand{\an}[1]{\textcolor{purple}{(AN: #1)}}

\newcommand{\dm}[1]{\textcolor{green}{(DM:#1)}}
\newcommand{\ignore}[1]{}
\newcommand{\newstuff}[1]{#1}

\usepackage[utf8]{inputenc} % allow utf-8 input
\usepackage[T1]{fontenc}    % use 8-bit T1 fonts
\usepackage{hyperref}       % hyperlinks
\usepackage{url}            % simple URL typesetting
\usepackage{booktabs}       % professional-quality tables
\usepackage{amsfonts}       % blackboard math symbols
\usepackage{nicefrac}       % compact symbols for 1/2, etc.
\usepackage{microtype}      % microtypography
\usepackage{xcolor}         % colors

\title{Computing All Optimal Partial $p$-Wasserstein Matchings on the Line}

\author{%
  Sebastian Angrick\\
  Karlsruhe Institute of Technology\\
  \texttt{sebastian.angrick@kit.edu} \\
    \And
  Jacobus Conradi \\
  University of Copenhagen \\
  \texttt{jaco@di.ku.dk} \\
    \And
  Mónika Csikós \\
  Université Paris Cité, IRIF \\
  \texttt{csikos@irif.fr} \\
    \And
  Niko Hastrich \\
  ETH Zurich\\
  \texttt{niko.hastrich@inf.ethz.ch} \\
    \And
  Danny Mittal \\
  University of Maryland, College Park\\
  \texttt{dannymittal@gmail.com} \\
    \And
  André Nusser \\
  Université Côte d'Azur, CNRS, Inria\\
  \texttt{andre.nusser@cnrs.fr}\\
    \And
  Krzystof Onak \\
  Boston University\\
  \texttt{konak@bu.edu} \\
    \And
  Sharath Raghvendra \\
  North Carolina State University\\
  \texttt{skraghve@ncsu.edu} \\
}

\begin{document}

\maketitle

\begin{abstract}
%\todo{Add funding to the Acknowledgement section at the end of the paper!}
For $p \ge 1$, the $p$-Wasserstein distance measures the minimum cost of
transporting probability mass between distributions, where moving unit mass between two points costs the
$p$th power of their distance. For discrete distributions in one dimension, full transport is especially
simple: after sorting, mass is matched in order along the line. By contrast,
partial and unbalanced transport on the line remains much less understood.
%Recently, Chapel and Tavenard~[ICLR'25] showed that, for $p=1$,  all optimal partial transport plans between distributions supported on $n$
%points, with each point carrying uniform mass,
%can be computed in $O(n\log n)$ time, using the metric property of distances. For $p>1$, these properties no longer apply and all existing
%approach requires $\Omega(n^2)$ time. Our main contribution is an FFT-based
%data structure for balanced-interval transport queries, which bypasses this
%quadratic bottleneck and yields an $O(p\,n\log^2 n)$-time algorithm for
%computing all optimal partial transports on the line for every finite
%$p\ge 1$. We also provide an open-source \texttt{C++} implementation that
%outperforms the state-of-the-art baseline on a range of synthetic instances.
Recently, Chapel and Tavenard~[ICLR'25] showed that, for $p=1$, all optimal
partial transport plans between distributions supported on $n$ points, with
uniform mass at each point, can be computed in $O(n\log n)$ time by exploiting
the metric structure of the cost. For $p>1$, this structure no longer applies,
and existing approaches require $\Omega(n^2)$ time. Our main contribution is an
FFT-based data structure for balanced-interval transport queries, which bypasses
this quadratic bottleneck and yields an $O(p\,n\log^2 n)$-time algorithm for
computing all optimal partial transports on the line for every finite
$p\ge 1$. We also provide an open-source \texttt{C++} implementation that
outperforms the state-of-the-art baseline on a range of synthetic instances.
%Recently, Chapel and Tavenard [ICLR'25] showed that for $p=1$, all optimal partial transport plans between uniform distributions on $n$ points can be computed in $O(n\log n)$ time. We extend their guarantee to all finite $p\geq1$: using an FFT-based balanced-interval data structure to bypass the quadratic bottleneck of naively evaluating augmenting-path costs, we compute all optimal partial transports on the line in $O(p\,n\log^2 n)$ time. We also provide an open-source \texttt{C++} implementation that outperforms the state-of-the-art baseline on a range of synthetic instances.
Finally, we establish a conditional lower bound for $p=\infty$: any subquadratic-time algorithm for computing all optimal partial transport plan costs on the line would violate the $(\min,+)$-Convolution Hypothesis. This separates the problem from full optimal transport, which is solvable in $O(n\log n)$.
%\todo{bibliography does not show arxiv/DOIs... maybe switch to plainurl}
\end{abstract}
\section{Introduction}

%\todo{mention asymmetric optimal transport ($|R|\neq |B|$)}

Optimal transport characterizes the difference between probability distributions $\mu$ and $\nu$ by measuring the minimum transportation cost with which we can transform $\mu$ into $\nu$. % with minimum transport of mass.
One particular flavor of optimal transport, the $p$-Wasserstein distance, is defined as the $p$th root of the minimum transportation cost needed
to move mass from $\mu$ to $\nu$, where transporting an amount of mass $\delta$
between points $a$ and $b$ incurs cost $\delta\, d(a,b)^p$; where $d(a,b)$ is the distance between \hbox{$a$ and $b$}. %The general setting is often referred to as optimal transport.
This metric is used widely for comparing distributions and point sets~\cite{rubner2000earth}.
Its appeal stems from its mathematical flexibility, its sensitivity to the
underlying geometry, and the fact that optimal transport plans provide explicit
correspondences between points, aiding interpretability and model transparency.
These features have made the $p$-Wasserstein distance a central tool in machine
learning and statistics~\cite{villani2009optimal,peyre2019computational}.
%\todo{Such a sentence deserves more citations, no? \url{https://www.annualreviews.org/content/journals/10.1146/annurev-statistics-030718-104938}}

In many practical settings the data may be noisy, contain outliers, or exhibit only partial overlap, making it undesirable to transport all available mass.
This has motivated substantial interest in the \emph{partial} and \emph{unbalanced} variants of optimal transport.
In these settings, the problem is not only to determine the transport plan, but also to decide how much mass should participate in the transport, while identifying and excluding points that are better treated as outliers.
A natural way to search for the ideal trade-off is through the \emph{optimal transport
profile} (or OT-profile), which maps each fraction $\alpha \in [0,1]$ to the
minimum cost of a transport plan that moves total \hbox{mass $\alpha$.} OT-profiles
have been used in applications such as noise filtering and positive-unlabeled learning~\cite{chapel2020partial,phatak2023computing,mukherjee2021outlier} and are a key component in the definition of the robust partial Wasserstein
metric~\cite{raghvendra2024new}, a robust variant of the
$p$-Wasserstein distance.
%that generalizes the $p$-Wasserstein distance, the
%L\'evy--Prokhorov distance~\cite{prokhorov1956convergence}, and total variation
%distances while providing improved robustness to noise.
%They also play a central role in the robust partial Wasserstein (RPW) distance~\cite{raghvendra2024new}, a recently proposed family of robust transport-based metrics. By varying its parameters, RPW recovers or interpolates between several classical distances, including total variation, p-Wasserstein, and, in a limiting/special-parameter case, the L\'evy--Prokhorov distance~\cite{prokhorov1956convergence}. This makes RPW a unifying framework while improving robustness to noise and outliers.

%In general, the OT-profile may have high representation complexity. In this
%paper, we focus on the discrete uniform setting, where each distribution is
%supported on $n$ points and each point carries mass $1/n$. In this setting,
%transporting mass $k/n$ corresponds to computing a minimum-cost matching of
%cardinality $k$ between the two point sets, which we refer to as a
%$k$-matching. Consequently, the OT-profile is fully determined by the sequence
%of costs $C_1,\ldots,C_n$, where $C_k$ denotes the minimum cost of a
%$k$-matching.

In general, the OT-profile may have high representation complexity. We focus on
the discrete uniform setting where each distribution is supported on $n$ points
of mass $1/n$.
In this setting, transporting mass $k/n$ amounts to selecting
$k$ vertex-disjoint pairs, each consisting of one point from each distribution; we call
such a set of pairs a \emph{$k$-matching}. The $p$-Wasserstein cost of a $k$-matching $M$ corresponds to the sum
of the transport costs over its pairs, namely $\sum_{(a,b)\in M} d(a,b)^p$.
Thus, the OT-profile is fully determined by the sequence of costs
$C_1,\ldots,C_n$, where $C_k$ is the minimum cost of any $k$-matching.

From an algorithmic perspective, computing a full optimal transport plan and
computing the OT-profile have the same best-known asymptotic complexity in
general settings. For arbitrary real-valued costs, both can be computed in
$O(n^3)$ time: the Hungarian algorithm~\cite{km_hungarian} constructs a minimum cost full matching by
successively computing minimum cost $k$-matchings for
$k=1,\ldots,n$, and therefore produces the entire OT-profile as a byproduct~\cite{phatak2023computing}.
In low-dimensional geometric settings, the full transport plan as well as the OT-profile can be computed in near-quadratic time using
$O(n^2)$ queries to a dynamic bichromatic closest pair data structure~\cite{vaidya1989geometry,aes_sjc99}. Thus,
near-quadratic time is a natural baseline even for low-dimensional geometric instances.
However, such running times can still be prohibitive for large-scale optimal
transport problems.

The one-dimensional setting is a notable exception to this scalability
challenge. Here, the transport cost between two points $a$ and $b$ is simply
$|a-b|^p$, and full optimal transport can be computed by sorting the points and
matching them in order, giving a near-linear-time algorithm for all
\hbox{$p \ge1$}~\cite{villani2009optimal,peyre2019computational,santambrogio2015optimal}.
This simplification underlies scalable machine-learning methods such as the
sliced Wasserstein distance, which reduces high-dimensional transport to a
collection of one-dimensional transport problems obtained by projecting data
onto lines~\cite{peyre2019computational,rabin2011barycenter,nguyen2025introductionslicedoptimaltransport}.

This computational simplicity, however, does not automatically extend to
partial and unbalanced optimal transport \cite{bai2023sliced,DBLP:conf/iclr/ChapelT25}: even in one dimension, efficiently
computing the full OT-profile remains nontrivial. Recently,
\cite{DBLP:conf/iclr/ChapelT25} showed that for $p=1$, the entire OT-profile on
the line can be computed in $O(n\log n)$ time for distributions supported on
$n$ points with uniform mass, matching the complexity of computing full optimal
transport in this setting. For all $p>1$, however, their approach requires
$\Theta(n^2)$ time. This raises a natural question:
 
 \emph{For $p>1$, can the OT-profile be computed with
the same efficiency as full one-dimensional optimal transport, or does a quadratic barrier arise?}

The difficulty of larger values of $p$ can already be seen in the following simple example.
Given two point sets $B,R \subset \mathbb{R}$ with $B = \{b_1, \dots, b_n\}$ and $R = \{r_1, \dots, r_n\}$ with %, and all points of $B$ lie to the left of all points of $R$, i.e., 
$b_1 < \cdots < b_n < r_1 < \cdots < r_n$.
For $p=1$, an optimal $k$-matching pairs the last $k$ points of $B$ with the first $k$ points of $R$.
Since any bijection between $\{b_{n-k+1},\ldots,b_n\}$ and $\{r_1,\ldots,r_k\}$ has the same $1$-Wasserstein cost, we can impose the following nested structure on the matching: $M_k = \{(b_{n-i+1}, r_i) : i = 1,\ldots,k\}$.
Moving from $M_k$ to $M_{k+1}$ simply requires adding the edge $(b_{n-k}, r_{k+1})$ without modifying existing pairs.
Hence, across all $k$ only $O(n)$ distinct edges appear.

For $p>1$, convexity breaks this degeneracy: the optimal pairing is now unique.
Even though the min-cost $k$-matching still matches the last $k$ points of $B$ to the first $k$ points of $R$, we now have to match them in order, i.e., the min-cost $k$-matching is $M_k = \{(b_{n-k+i}, r_i) : i = 1,\ldots,k\}$.
Note that increasing the cardinality from $k$ to $k+1$ forces every edge to shift in order to obtain the matching $M_{k+1} = \{(b_{n-(k+1)+i}, r_i) : i = 1,\ldots,k+1\}$.
Thus $M_k$ and $M_{k+1}$ share \emph{no} edges. The sequence $M_1, M_2, \ldots, M_n$ therefore contains $\Theta(n^2)$ distinct edges.
Consequently, any approach that explicitly maintains the partial matchings
incurs quadratic total complexity for $p>1$. %For $p=\infty$, we use this
\subsection{Our Contribution.}
In this work, we show that for any fixed $p \in \N$, the matching profile under the $p$-Wasserstein cost can be computed in $O(p\,n\log^2 n)$ time. 
We complement this upper bound with an open-source \texttt{C++} implementation that outperforms the state-of-the-art on a broad range of synthetic instances.
Finally, for $p=\infty$, we establish a conditional lower bound ruling out truly subquadratic algorithms. % unless a popular fine-grained complexity conjecture fails.

{\bf Upper Bounds.}
We show the following theorem, constituting our main result.

\begin{restatable}[$k$-partial $p$-Wasserstein Distance]{theorem}{maintheorem}
\label{thm:algo-partial-wasserstein}
Let $B, R \subset \mathbb{R}$ be two point sets of size $n$. For any $p \in \mathbb{N}$, 
the $k$-partial $p$-Wasserstein distance between $B$ and $R$ can be computed for all $k \in [n]$ in a total of 
$O(p\,n \log^{2} n)$ time using $O(n \log n)$ space.
\end{restatable}

Similar to the approach of~\cite{DBLP:conf/iclr/ChapelT25}, our algorithm can be viewed as a one-dimensional specialization of the Hungarian algorithm.
It maintains a collection of at most $n$ disjoint balanced intervals, each
containing an equal number of red and blue points, that compactly represents
the current minimum-cost partial matching. Interestingly, the minimum-cost
$(k+1)$-matching can be obtained from the minimum-cost $k$-matching by modifying
only few intervals. Thus, although all partial matchings together
may contain $\Theta(n^2)$ distinct edges, their structure admits an
$O(n)$-space representation.

The main remaining task is to identify which intervals to modify. This reduces
to evaluating minimum matching costs inside balanced intervals, which is handled
by our main technical contribution, summarized by the following theorem.

\begin{theorem}\label{thm:main-balanced-interval-ds}
Let $B,R \subset \mathbb{R}$ be point sets of size $n$ each. For any
$p \in \mathbb{N}$, in $O(p\,n \log^2 n)$ time we can construct a
data structure of size $O(n \log n)$ that, given a balanced query interval $I$ (i.e., $|I\cap R|=|I\cap B|$) returns the $p$-Wasserstein cost of the minimum-cost matching between $I\cap R$ and $I\cap B$ in $O(\log n)$ time.
\end{theorem}

%\Cref{thm:main-balanced-interval-ds} relies on the following
%The optimal matching inside any balanced interval always pairs the left-most red point with the left-most blue point.
%Thus, if the first red point in the interval has index $k$ and the first blue point has index $k+t$, then every edge in the optimal matching inside that interval has shift $t$ (i.e., the indices of the matched endpoints differ by exactly $t$).
The optimal matching inside any balanced interval pairs red and blue points in
order. Hence, if the first red point in the interval has index $k$ and the first
blue point has index $k+t$, then every matched pair has the same \emph{shift}
$t$.
The data structure of Theorem~\ref{thm:main-balanced-interval-ds} exploits this fact
with a one-dimensional range tree. Each node corresponding to a contiguous range of points $S\subset R\cup B$ stores
matching-cost information for all $O(|S|)$ relevant shifts between red and blue
points in $S$. Given a balanced query interval, we determine its shift $t$ from
its leftmost red and blue points, and recover its optimal matching cost by
aggregating the precomputed shift-$t$ values over the $O(\log n)$ maximal range-tree
nodes into which the interval decomposes.
This aggregation is nontrivial: matching costs are not generally decomposable
over disjoint subproblems. Here, decomposability follows from the
one-dimensional order structure together with balancedness; see
Lemma~\ref{lem:crucial}. The remaining challenge is preprocessing all
shift-values without quadratic work. We encode the costs for all shifts at a
node as coefficients of carefully chosen polynomials; expanding $|a-b|^p$ and
using FFT computes these values efficiently. Summed over the range tree, this
gives $O(p\,n\log^2 n)$ preprocessing time and $O(\log n)$ query time.

{\bf Implementation.}
%We complement our theoretical results with an open-source \texttt{C++}-implementation of our algorithm. To assess its practical performance, we
%compare it against the state-of-the-art implementation of
%\cite{DBLP:conf/iclr/ChapelT25} on a broad collection of synthetic instances.
%%Our experiments show that our implementation is on par, and often consistently faster across the
%entire tested regime. This speedup persists even on instances that are
%particularly favorable to the baseline, where the balanced intervals considered
%during the algorithm tend to be small. Moreover, as the problem size grows, our
%implementation exhibits substantially better scaling behavior, demonstrating
%that the theoretical improvements translate into clear practical gains.
We complement our theoretical results with an open-source \texttt{C++} implementation.
% Compared with the state-of-the-art implementation of \cite{DBLP:conf/iclr/ChapelT25} on a broad collection of synthetic instances, our implementation is competitive throughout and often substantially faster, with increasingly bigger speedups as the input size grows.
Compared with the state-of-the-art implementation of \cite{DBLP:conf/iclr/ChapelT25} on a broad collection of synthetic instances, our implementation is competitive throughout and orders of magnitude faster on most large inputs, with increasingly bigger speedups as the input size grows.

{\bf Lower Bounds.}
%Our upper bound applies to every finite $p \in \mathbb{N}$, but it does not
%extend to the limiting case $p=\infty$. For this case, we prove a conditional
%lower bound: assuming the $(\min,+)$-Convolution
%Conjecture~\cite{DBLP:journals/talg/CyganMWW19}, no truly subquadratic-time
%algorithm can compute all $k$-partial $\infty$-Wasserstein matching costs.
%This establishes a separation between the computational complexity for full transport plan and all optimal partial transports. Indeed,
%full $\infty$-Wasserstein transport admits subquadratic exact algorithms already
%in two dimensions, whereas computing the entire partial transport profile is
%conditionally quadratic-hard even in the highly restricted one-dimensional
%setting.
Our upper bound applies to every finite $p \in \mathbb{N}$, but not to
$p=\infty$. For this case, assuming the $(\min,+)$-Convolution
Conjecture~\cite{DBLP:journals/talg/CyganMWW19}, we prove that computing all
$k$-partial $\infty$-Wasserstein matching costs cannot be performed in truly subquadratic time.
This separates full transport from the OT-profile: full
$\infty$-Wasserstein transport admits subquadratic exact algorithms even in
two dimensions, whereas the complete partial transport profile is hard already in one dimension.

\begin{restatable}[]{theorem}{minpluslowerbound}\label{thm:firstLowerbound}
Assume the \((\min, +)\)-Convolution Conjecture and consider \(B\) and \(R\) to be sets of \(n\) integer points in \([-W, W]\).
Then, there is no algorithm that takes \(B\) and \(R\) as input and outputs all \(k\)-partial \(\infty\)-Wasserstein matching costs in time \(O(n^{2-\varepsilon}\polylog W)\) for any \(\varepsilon>0\).
\end{restatable}

In \Cref{sec:omitted_lbs}, we give two additional lower bounds for a related data-structure
setting: given $p \in \mathbb{N}$ and point sets $B,R\subset\mathbb{R}$, answer
queries that for intervals $I_B,I_R\subset\mathbb{R}$ with
$|B\cap I_B|=|R\cap I_R|$ return the $p$-Wasserstein distance between
$B\cap I_B$ and $R\cap I_R$. Interestingly, this setting is harder than for a single query interval, unless an almost-quadratic matrix multiplication algorithm exists.
%In \Cref{sec:omitted_lbs} we present two more lower bounds for a very related setting, in which we the goal is the construction of a data structure that is given $p \in \mathbb{N}$ and point sets $B, R \subset \mathbb{R}$, and must answer queries of the following form: 
%Given two intervals $I_B, I_R \subset \mathbb{R}$ with $|B \cap I_B| = |R \cap I_R|$, return the $p$-Wasserstein distance between $B \cap I_B$ and $R \cap I_R$. Surprisingly, this setting is provably harder.
%\todo{The part from here on seems to have too much space considering that it is entirely in the appendix.}

\section{Data Structure for Balanced Interval Matching}
\label{sec:data-structure}

In this section we describe the data structure from Theorem~\ref{thm:main-balanced-interval-ds} supporting queries of the form: Given a balanced interval $I$, i.e., an interval such
that $|B \cap I| = |R \cap I|$, return the cost of the minimum-cost complete matching between $B \cap I$ and $R \cap I$. We denote this value by $w^p(I)$.
\newstuff{This data structure constitutes our main technical contribution and allows us, using the algorithm described in \Cref{sec:algorithm}, to efficiently compute minimum-cost $k$-matchings for all $k = 1, ..., n$ and $p \geq 1$}.

%\begin{theorem}\label{thm:balanced-interval-ds}
%Let $B,R \subset \mathbb{R}$ be point sets of size $n$ each, and let
%$p \in \mathbb{N}$ be fixed. In $O(n \log^2 n)$ time we can construct a
%data structure of size $O(n \log n)$ that returns $w^p(I)$ for any
%balanced interval $I$ in $O(\log n)$ time.
%\end{theorem}

%Section~\ref{sec:tree-structure} describes the details of our data structure.
%%Section~\ref{sec:query} presents the query algorithm and its correctness.
%Section~\ref{sec:preprocessing} describes the preprocessing algorithm.

%\subsection{Our Data Structure}

{\bf The overall structure.}
Let $r_1,\dots,r_n$ and $b_1,\dots,b_n$ be the red and blue points in non-decreasing order, and let $p_1,\dots,p_{2n}$ denote the merged sorted order of all points. Our data structure is a standard binary range tree $\mathcal{T}$ over the index range $[1,2n]$. Each node $S=[a,b]$ has midpoint $\mathsf{mid}_S=\lfloor(a+b)/2\rfloor$, left child $\leftChild{S}=[a,\mathsf{mid}_S]$, and right child $\rightChild{S}=[\mathsf{mid}_S+1,b]$. A node $S=[a,b]$ corresponds to contiguous ranges $r_j,\ldots,r_{j'}$ and $b_k,\ldots,b_{k'}$ of red and blue points.
To keep notation lightweight, we also write $r,b \in S$ for $r \in R$ and $b \in B$ to denote that their global indices in the merged order of the points lie in the index range $S$.

For a pair $(r_j,b_k)$, we call $k-j$ its \emph{shift}. Thus, the shifts of pairs contained in $S$ lie in $[\,k-j',\,k'-j\,]$, which we call the range of \emph{relevant shifts} for $S$. For a nonempty balanced interval $I$ with leftmost red and blue points $r_j$ and $b_k$, the \emph{shift of $I$} is $k-j$; since the optimal matching in $I$ matches its red and blue points in order, all its pairs have this shift. Refer to \Cref{fig:relevantshift}.

We augment each node $S=[a,b]$ of the tree with $O(|S|)$ values: For each relevant shift $d$, we store
\begin{compactitem}
%     \item $W(\leftChild{S},\rightChild{S},d)=\sum_{(r_j,b_{j+d})\in \leftChild{S}\times \rightChild{S}}|r_j-b_{j+d}|^p$, the total cost of all pairs $(r_j,b_{j+d})$ at shift $d$ with $r_j$ in $\leftChild{S}$ and $b_{j+d}$ in $\rightChild{S}$;
%     \item $W(\rightChild{S},\leftChild{S},d)=\sum_{(r_j,b_{j+d})\in \rightChild{S}\times \leftChild{S}}|r_j-b_{j+d}|^p$, the total cost of all pairs $(r_j,b_{j+d})$ at shift $d$ with $r_j$ in $\rightChild{S}$ and $b_{j+d}$ in $\leftChild{S}$;
%     \item $W(S,d)=\sum_{(r_j,b_{j+d})\in S\times S}|r_j-b_{j+d}|^p$ the total cost of pairs contained in $S$ at shift $d$.
     \item the total cost of all pairs $(r_j,b_{j+d})$ at shift $d$ with $r_j$ in $\leftChild{S}$ and $b_{j+d}$ in $\rightChild{S}$: $W(\leftChild{S},\rightChild{S},d) \coloneqq \sum_{(r_j,b_{j+d})\in \leftChild{S}\times \rightChild{S}}|r_j-b_{j+d}|^p$, and
     \item the total cost of all pairs $(r_j,b_{j+d})$ at shift $d$ with $r_j$ in $\rightChild{S}$ and $b_{j+d}$ in $\leftChild{S}$: $W(\rightChild{S},\leftChild{S},d) \coloneqq \sum_{(r_j,b_{j+d})\in \rightChild{S}\times \leftChild{S}}|r_j-b_{j+d}|^p$, and
     \item the total cost of pairs contained in $S$ at shift $d$: $W(S,d) \coloneqq \sum_{(r_j,b_{j+d})\in S\times S}|r_j-b_{j+d}|^p$.
\end{compactitem}
For shifts $d$ that are not relevant for $S$, the values are defined to be zero, as for any pair $(r_j,b_{j+d})$ at shift $d$ not both $r_j$ and $b_{j+d}$ can be in $S$. These values we store implicitly. As an aside, we remark that for any shift $d$ at least one of $W(\leftChild{S},\rightChild{S},d)$ and $W(\rightChild{S},\leftChild{S},d)$ is zero.
As each node $S$ stores $O(|S|)$ values, the total size of the data structure is in $O(n\log n)$.

We now first show how we can answer the desired queries using this data structure to then show how we can actually build the data structure with the correct costs.

\begin{figure}
\centering
\includegraphics[width=0.45\textwidth]{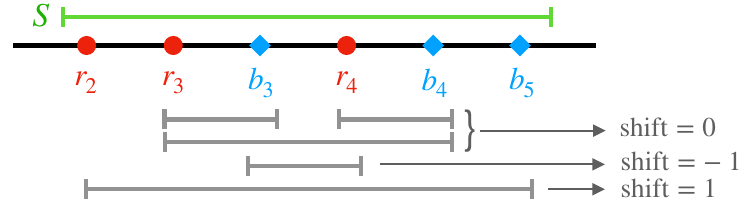}
\caption{Balanced intervals in $S$ and their shift values. The range of relevant shifts is $[-1,3]$.%\todo[inline]{Figure not referenced in text}
}
\label{fig:relevantshift}
\end{figure}

{\bf Answering Queries.} %\todo{This descrption currently feels a bit lack luster}
%\todo{Somewhere (here?) it should be made more clear that we first determine the shift to then compute the cost with that shift.}
Let $I \subseteq [1,2n]$ be a \newstuff{nonempty} balanced interval, and let $d$ be the shift of its optimal matching.\footnote{\newstuff{As a preliminary step in the query algorithm, $d$ must be determined. This can be done easily by maintaining a binary search tree over each of the red and blue points, then given $I$, querying said binary search tree to determine the leftmost red point $r_j$ in $I$ and leftmost blue point $b_k$ in $I$, after which we can simply take $d = k - j$.}} 
%Querying the data structure for $w^p(I)$ now consists of 
For any node $S \in \mathcal{T}$, we define the cost $\mathrm{Cost}(I\cap S, d)$ of all pairs in $S\cap I$ at shift $d$ via the %procedure that computes the contribution at shift $d$ of all
%red-blue pairs contained in $I \cap S$. This quantity is computed
%recursively using the precomputed tables $W(\cdot,\cdot)$ stored at each
%node.
following recursive procedure (refer to \Cref{fig:decomposition}):

\begin{compactitem}
    \item If $S \subseteq I$, return $W(S,d)$.
    \item If $I$ intersects both children of $S$, return
    \[W(\leftChild{S},\rightChild{S},d)
        + W(\rightChild{S},\leftChild{S},d)
        + \mathrm{Cost}(I\cap \leftChild{S},d)
        + \mathrm{Cost}(I\cap \rightChild{S},d).\]
    \item If $I\cap S \subseteq \leftChild{S}$, return $\mathrm{Cost}(I\cap \leftChild{S}, d)$.
    \item If $I\cap S \subseteq \rightChild{S}$, return $\mathrm{Cost}(I\cap \rightChild{S}, d)$.
\end{compactitem}
To answer a query, we compute $\mathrm{Cost}(I\cap S_0,d)$, where $S_0=[1,2n]$ is the root of the tree $\mathcal{T}$.
We now prove that this recursive definition indeed results in the values that we want to report.

\begin{lemma}
\label{lem:crucial}
    Let $I\subseteq[1,2n]$ be a nonempty balanced interval, and let $d$ be the shift of its optimal matching. Then $\mathrm{Cost}(I\cap S_0,d)=w^p(I)$, where $S_0=[1,2n]$ is the root of $\mathcal{T}$.
\end{lemma}
\begin{proof}
We show that every pair of shift $d$ contributing to the optimal matching inside $I$ is counted exactly once by the recursion, and that no other pair is counted.

The only point requiring justification is the use of the full terms $W(\leftChild{S},\rightChild{S},d)$ and $W(\rightChild{S},\leftChild{S},d)$ when $I$ intersects both children of a node $S$. We claim that in this case every shift-$d$ pair crossing between $\leftChild{S}$ and $\rightChild{S}$ is fully contained in $I$.
Indeed, suppose \(R\cap I=\{r_\alpha,\ldots,r_{\alpha+m-1}\}\) and \(B\cap I=\{b_\beta,\ldots,b_{\beta+m-1}\}\) for some $\alpha$, $\beta$, and $m$. Since $I$ is balanced, the optimal matching inside $I$ matches points in order, so its shift is $d=\beta-\alpha$. Now consider a shift-$d$ pair $(r_j,b_{j+d})$ with $r_j\in \leftChild{S}$ and $b_{j+d}\in \rightChild{S}$. If $j<\alpha$, then $j+d<\beta$, so $b_{j+d}$ lies before the first blue point of $I$ and hence before $I$ in the merged order. This contradicts $b_{j+d}\in \rightChild{S}$, since $I$ intersects both children of $S$. Similarly, if $j>\alpha+m-1$, then $r_j$ lies after the last red point of $I$ and hence after $I$ in the merged order, contradicting $r_j\in \leftChild{S}$. Hence $\alpha\le j\le \alpha+m-1$, and therefore also $\beta\le j+d\le \beta+m-1$. 
Thus both points of the pair $(r_j, b_{j+d})$ lie in $I$.
The case of a pair from $\rightChild{S}$ to $\leftChild{S}$ is symmetric.

Now consider the recursion at a node $S$. If $S\subseteq I$, the procedure returns $W(S,d)$, which is exactly the total cost of all shift-$d$ pairs contained in $S$. If $I$ intersects both children of $S$, then every shift-$d$ pair in $I\cap S$ is of exactly one of the following four types: contained in $\leftChild{S}$, contained in $\rightChild{S}$, crossing from $\leftChild{S}$ to $\rightChild{S}$, or crossing from $\rightChild{S}$ to $\leftChild{S}$. The first two types are handled recursively, and the last two are counted by $W(\leftChild{S},\rightChild{S},d)$ and $W(\rightChild{S},\leftChild{S},d)$. By the arguments above, these cross terms count no pairs outside $I$. If $I$ intersects only one child, the recursion only descends into that child, so again no pair is lost or counted twice.

Thus, every shift-$d$ pair contained in $I$ is counted exactly once,and no shift-$d$ pair outside $I$ is counted. Since $d$ is the shift of the optimal matching inside the balanced interval $I$, these are precisely the pairs of that matching. Finally, as $I\subseteq S_0=[1,2n]$, we have 
\(\mathrm{Cost}(I\cap S_0,d)=w^p(I)\).
\end{proof}

Since the recursion follows the standard range-query traversal in a binary tree over $[1,2n]$, it has at most two active root-to-leaf search paths. All other visited nodes are fully contained in $I$ and terminate immediately. Hence the recursion visits $O(\log n)$ nodes and spends $O(1)$ time at each node. Thus, the query time is $O(\log n)$.

\begin{figure}
    \centering
    \includegraphics[width=\textwidth]{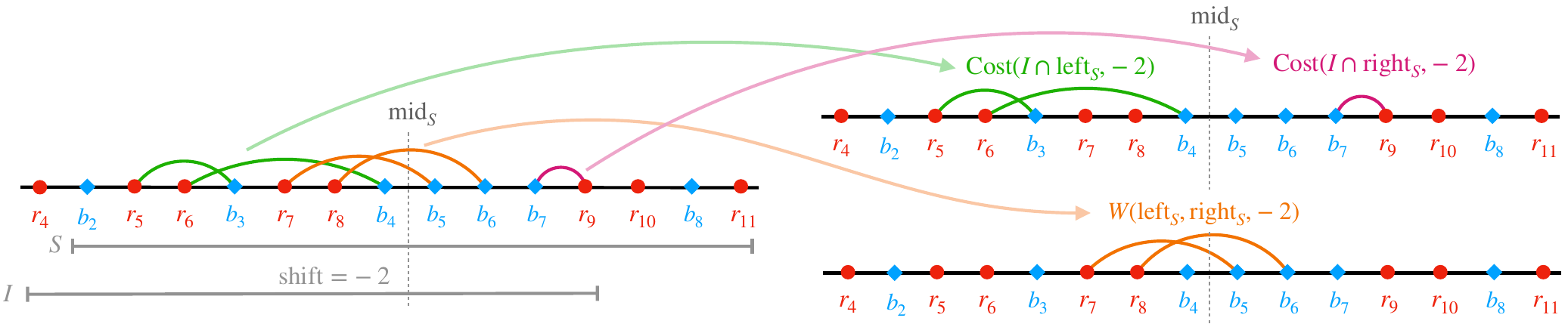}
    \caption{The optimal matching of $I$ has shift $-2$. We decompose the edges of the optimal matching for $I$ that are contained in $S\cap I$ into $3$ disjoint matchings of cost $\mathrm{Cost}(I \cap \mathrm{left}_S,d)$, $\mathrm{Cost}(I \cap \mathrm{right}_S,d)$, and $W(\mathrm{left}_S,\mathrm{right}_S,d)$. Note that $W(\mathrm{right}_S,\mathrm{left}_S,-2) =0$. 
    }
    \label{fig:decomposition}
\end{figure}

{\bf Constructing the data structure.}
%\label{sec:preprocessing}
%The preprocessing step computes all values
%$W(\mathrm{left}_S,\mathrm{right}_S,d)$,
%$W(\mathrm{right}_S,\mathrm{left}_S,d)$, and $W(S,d)$ for every node
%$S \in \mathcal{T}$. The two cross terms are symmetric, so we describe
%only the computation of $W(\mathrm{left}_S,\mathrm{right}_S,d)$.
Let $S = [a,b]$ be a node of the range tree. Suppose the red points in $\mathrm{left}_S$ are $r_u,\dots,r_v$ in the global red order, and the blue points in $\mathrm{right}_S$ are $b_s,\dots,b_t$ in the global blue order. For any shift $d$, the value $W(\mathrm{left}_S,\mathrm{right}_S,d)$ is the sum of $|b_{j+d} - r_j|^p$ over all $j$ with $j \in [u,v]$ and $j+d \in [s,t]$. Since all red points in $\mathrm{left}_S$ lie to the left of all blue points in $\mathrm{right}_S$, we always have $b_{j+d} > r_j$, so $|b_{j+d} - r_j|^p = (b_{j+d} - r_j)^p$.
%We compute these values simultaneously for all relevant shifts $d$ by
%expanding each term using the binomial identity
Expanding $(b_{j+d} - r_j)^p$, we obtain \((b_{j+d} - r_j)^p = \sum_{c=0}^p \binom{p}{c} (-1)^c r_j^c\, b_{j+d}^{\,p-c}\), and hence 
\[W(\mathrm{left}_S,\mathrm{right}_S,d)=\sum_{\substack{j\in[u,v]\\ j+d\in[s,t]}}(b_{j+d} - r_j)^p=\sum_{c=0}^p \binom{p}{c}(-1)^c\left(\sum_{\substack{j\in[u,v]\\ j+d\in[s,t]}}r_j^c\, b_{j+d}^{\,p-c}\right).\]
Thus, to obtain $W(\mathrm{left}_S,\mathrm{right}_S,d)$, it suffices to compute, for each $c \in \{0,\dots,p\}$, the sums
\[\sum_{\substack{j\in[u,v],\, k\in[s,t]\\ k-j=d}} r_j^c\, b_k^{p-c}\]
for all relevant shifts $d$. These sums are entries of a convolution: for $A=(r_u^c,r_{u+1}^c,\ldots,r_v^c)$ and $B=(b_t^{p-c},b_{t-1}^{p-c},\ldots,b_s^{p-c})$, the sum above is the entry $(A*B)[t-u-d]$, where $*$ denotes the discrete convolution. Hence, for all relevant shifts $d$ these values can be computed simultaneously by one convolution in $O((b-a)\log(b-a))$ time via FFT \cite{oppenheim1999discrete,cooley1965algorithm}. Repeating this for every $c=0,\ldots,p$, the values $W(\mathrm{left}_S,\mathrm{right}_S,d)$ for all relevant shifts are computed in $O(p(b-a)\log(b-a))$ time. For more details, refer to \Cref{appendix:fft}. We compute $W(\rightChild{S},\leftChild{S},d)$ symmetrically.
%Since the right-left computation is
As the sum of the lengths of all node ranges in $\mathcal{T}$ is $O(n\log n)$, the total construction time is $O(pn\log^2 n)$.

Finally, the values $W(S,d)$ are computed bottom-up using the recurrence
\[ W(S,d)
    = W(\mathrm{left}_S,d)
    + W(\mathrm{left}_S,\mathrm{right}_S,d)
    + W(\mathrm{right}_S,\mathrm{left}_S,d)
    + W(\mathrm{right}_S,d),\]
which requires $O(|S|)$ time per node $S$ and $O(n\log n)$ time overall, implying \Cref{thm:main-balanced-interval-ds}. %This
\vspace{-1mm}
\section{A Near-Linear Time Algorithm for $k$-Partial $p$-Wasserstein Distance}
\vspace{-1mm}
\label{sec:algorithm}
Using the data structure introduced in \Cref{sec:data-structure}, we now present an $O(pn \log^2 n)$-time algorithm to compute minimum-cost $k$-matchings for all $k=1,\ldots,n$ on the line. We note that this general algorithmic framework follows the classical Hungarian algorithm applied to points on a line. A similar approach has previously been used by, e.g., \cite{DBLP:conf/iclr/ChapelT25}. %For the sake of completeness however, we restate and reprove its necessary properties.
%In this section, we describe an efficient algorithm for computing minimum-cost $k$-matchings for all $k=1,\ldots,n$, assuming access to a data structure from Theorem~\ref{thm:main-balanced-interval-ds}. \newstuff{The overall approach described in this section follows the classical Hungarian algorithm applied to points on a line. However, for the sake of readability, we repeat and reprove its necessary properties.}
We begin by reviewing some of the basic properties of matchings needed to describe our algorithm.

A vertex $v \in R \cup B$ is said to be \emph{free} with respect to a matching $M$ if no edge of $M$ is incident to $v$. Let $R_F$ and $B_F$ denote the sets of free red and blue points, respectively. A path $P$ in the complete bipartite graph on $R \cup B$ is called \emph{alternating} (with respect to $M$) if its edges alternate between those in $M$ and those not in $M$. An alternating path is \emph{augmenting} if its two endpoints are free vertices, one in $R_F$ and the other in $B_F$. Flipping a matching along such a path removes the matching edges in $P$ and adds the non-matching edges of $P$; when $P$ is augmenting, this operation is called an \emph{augmentation}, and it increases the size of the matching by one.

For any alternating path or cycle $P$ with respect to $M$, we define its \emph{net cost} to be
\[\Phi(P)\;=\;\sum_{(a,b) \in P \setminus M} |a-b|^p\;-\;\sum_{(a,b) \in P \cap M} |a-b|^p .\]
The net cost $\Phi(P)$ is exactly the change in the total cost of the matching when $M$ is replaced by the symmetric difference $M \oplus P$. The next lemma shows that augmenting a minimum-cost $k$-matching along the minimum net-cost augmenting path gives a minimum-cost $(k+1)$-matching.

\begin{lemma}
\label{lem:hungarianprop}
Let $M_k$ be a minimum-cost $k$-matching. Among all minimum-cost $(k+1)$-matchings, let $M_{k+1}$ be one that maximizes the number of edges it shares with $M_k$. Then the symmetric difference $M_k \oplus M_{k+1}$ consists of a single augmenting path with respect to $M_k$, and this path has minimum net cost among all augmenting paths for $M_k$.
\end{lemma}

We defer the proof to \Cref{sec:proof:hungarianprop}. The classical Hungarian algorithm begins with the empty matching and, using a primal-dual framework, repeatedly identifies the minimum net-cost augmenting path and augments the matching along it. In this way, it successively constructs minimum-cost $k$-matchings for all $k = 1,\ldots,n$.

Next, we introduce a compact representation of minimum-cost partial matchings on the line. This reduces computing the net-cost of an augmenting path to two queries to our data structure.

{\bf Interval Representation.} %\label{sec:interval-representation}
Consider a minimum-cost $k$-matching $M$ between the blue points $B$ and the red points $R$. Let $B_F \subseteq B$ and $R_F \subseteq R$ denote the free (unmatched) points with respect to $M$. Each matching edge $(b,r)\in M$ naturally defines the interval $I_{b,r} := [\min\{b,r\},\,\max\{b,r\}]$. Edges in the minimum-cost matching for points on the line satisfy the following property.

% \medskip
% \noindent\textbf{Observation 1.}
\begin{observation}\label{obs:1}
No free point lies in the interior of any interval $I_{b,r}$ 
%defined by an edge $(b,r)\in M$.
of an edge $(b,r)\in M$.
\end{observation}
% This follows from the fact that if 
If a free point $p \in (b,r)$ existed, replacing $(b,r)$ by $(p,r)$ or $(b,p)$ would strictly shorten the edge and reduce the total cost of $M$, contradicting its minimality.

Given a set of disjoint intervals $\mathcal{I}$, we say that two intervals $I, J \in \mathcal{I}$ are \emph{adjacent} if there is no other interval of $\mathcal{I}$ lying strictly between them in the ordering of the real line.
Let $U_M := \bigcup_{(b,r)\in M} \{I_{b,r}\}$ be the set of all edge-intervals, see \Cref{fig:matching-intervals} for an example.  %\todo{The previous description did not match with the picture...}
%\todo{I find the definition of $U_M$ together with the figure confusing. So, $U_M$ is just a set of points and $\mathcal{I}'_M$ is a set of intervals, but otherwise they represent the same thing?}
We describe the \emph{interval representation} $\mathcal{I}_M$ of $M$ in two steps. First, let $\mathcal{I}'_M$ be the set of disjoint intervals, one spanning each connected component of $\bigcup U_M$. We then obtain $\mathcal{I}_M$ by repeatedly merging adjacent intervals of $\mathcal{I}'_M$ as follows: if $I,J \in \mathcal{I}'_M$ are adjacent, with $I=[x_1,x_2]$ and $J=[x_3,x_4]$ and $x_2 \le x_3$, and the interval $[x_1,x_4]$ contains no free points, we replace $I$ and $J$ by the single interval $[x_1,x_4]$.  We continue this process until no further merge is possible. We refer to any pair of free blue and red points $(b,r)\in B_F\times R_F$ as an \emph{adjacent free pair} if no other free point is contained in $I_{b,r}$.

\begin{figure}
\centering
    \includegraphics[width=0.55\textwidth]{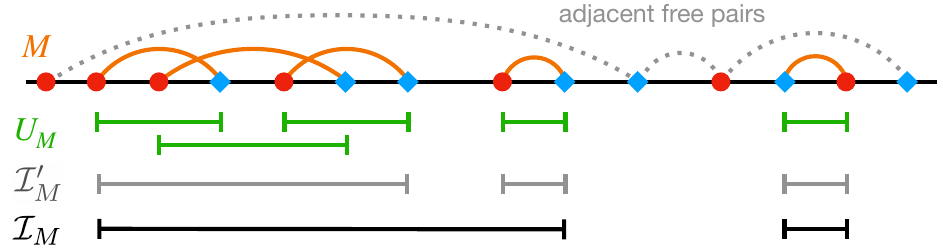}
    \caption{Construction of $\mathcal I_M$ and adjacent free pairs.}\label{fig:matching-intervals}
\end{figure}

\noindent
For any minimum-cost matching $M$, the set of intervals $\mathcal{I}_M$ satisfies the following properties:

\begin{compactenum}[I]
    \item\label{property:balance} \emph{Balancedness:}
    Every interval in $\mathcal{I}_M$ contains the same number of blue and red points. Indeed, the intervals in $\mathcal{I}'_M$, the connected components of $U_M$, are balanced by construction, and merging adjacent disjoint balanced intervals with no free point between them preserves balancedness.

    \item\label{property:char} \emph{Characterization of matched points:} A point $p$ is matched by $M$ if and only if $p$ lies in some interval of $\mathcal{I}_M$. All matched points lie in $U_M$, and by \Cref{obs:1} no free point lies in the interior of any edge-interval; the merging step also does not introduce free points.

    \item\label{property:monotone} \emph{Monotone structure within intervals:} Inside each interval $I \in \mathcal{I}_M$, the matching $M$ pairs the $i$-th smallest blue point with the $i$-th smallest red point. All matching edges lie within intervals of $\mathcal{I}_M$. If this monotone pairing failed inside some $I$, replacing the edges of $M$ in $I$ by the optimal matching on the points of $I$ would yield a cheaper matching, contradicting optimality of $M$.

    \item\label{property:free-pts} \emph{At most one interval between adjacent free points:} For every adjacent free pair $(b,r)$, the interval $I_{b,r}$ contains at most one interval of $\mathcal{I}_M$. Otherwise, since $I_{b,r}$ contains no free points, the contained intervals would have been merged.
\end{compactenum}

These properties make $\mathcal{I}_M$ a compact and convenient representation of the matching $M$. For any balanced interval $I$, let $M_I$ denote the minimum-cost matching on the points in $I$, and define $w^p(I) := \sum_{(b',r') \in M_I} |b' - r'|^p$ to be the matching cost. Using this notation, the net-cost of an augmentation step can be expressed entirely in terms of interval costs.

\begin{lemma} \label{lem:net_cost}
Let $M$ be a minimum-cost $k$-matching, and let $P$ be a minimum net-cost augmenting path with respect to $M$. Then $P$ has endpoints at an adjacent free pair $(b,r)$, and its net-cost is
\[\Phi(P) = \Phi(b,r) \coloneqq \begin{cases}
|b - r|^p, & \text{if } I_{b,r} \text{ contains no interval of } \mathcal{I}_M,\\[6pt]
w^p(I_{b,r}) - w^p(I), & \text{if } I\in\mathcal{I}_M \text{ lies in } I_{b,r}.\end{cases}\]
\end{lemma}
See \Cref{sec:proof_lem_net_cost} for the proof. We now present our algorithm, which uses this characterization to identify minimum net-cost augmenting paths using only two interval-cost queries.

{\bf The Algorithm.} %\label{sec:algo}
%In this section, we assume access to a data structure $\mathcal{D}$ from \cref{thm:main-balanced-interval-ds}.
%Using our data structure from \Cref{thm:main-balanced-interval-ds}, w
%We now describe our algorithm for computing the
%minimum-cost $k$-matchings for all $k$.
Our algorithm incrementally constructs the minimum-cost matchings of sizes $1,2,\ldots,n$. At all times it maintains the current matching $M$ through its interval representation $\mathcal{I}_M$, the free sets $B_F \subseteq B$ and $R_F \subseteq R$, and the set of adjacent free pairs together with their net-costs.

Initially, we set $M=\emptyset$, $\mathcal{I}_M=\emptyset$, $B_F=B$, and $R_F=R$. Each consecutive red-blue pair is adjacent, and since no intervals are present, each adjacent pair $(b,r)$ receives the initial weight $\Phi(b,r)=|b-r|^p$. Each adjacent pair is added to a min-heap $\mathcal{Q}$ with their weight as the key value. We build the data structure $\mathcal{D}$ from \Cref{sec:data-structure} on the point sets $R$ and $B$.

The algorithm now iteratively transforms a minimum-cost $k$-matching into a minimum-cost $(k+1)$-matching via the following steps.

\begin{compactenum}
    \item \textbf{Select the minimum-cost augmentation.} From the min-heap $\mathcal{Q}$, choose the adjacent free pair $(b,r)$ minimizing $\Phi(b,r)$.

    \item \textbf{Update the interval representation and free sets.} By Property~\ref{property:free-pts} of $\mathcal{I}_M$, at most one interval of $\mathcal{I}_M$ lies inside $I_{b,r}$; remove it if present, and insert $I_{b,r}$. Merge adjacent intervals in $\mathcal{I}_M$.
          %Apply the merging rule from Section~\ref{sec:interval-representation}.
          Finally, mark $b$ and $r$ as matched and remove them from $B_F$ and $R_F$.

    \item \textbf{Update adjacency.}
          Remove all adjacent pairs involving $b$ or $r$. At most one new adjacent free pair $(b',r')$ is created.

    \item \textbf{Update net-costs.}
          For any newly created adjacent free pair $(b',r')$, we query $\mathcal{D}$ to compute $\Phi(b',r')= w^p(I_{b',r'}) - w^p(I)$, where $I$ is the unique interval of $\mathcal{I}_M$ that lies inside $I_{b',r'}$.
         \ignore{ \[
          \Phi(b',r') =
          \begin{cases}
          |b'-r'|^p,
            & \text{if } I_{b',r'} \text{ contains no interval of } \mathcal{I}_M,\text{\mcs{Is this possible?}\sa{as a newly created interval, I think no. I still like it for the overview}}\\[6pt]
          w^p(I_{b',r'}) - w^p(I),
            & \text{if } I\in\mathcal{I}_M \text{ lies inside } I_{b',r'}.
          \end{cases}
          \]}
\end{compactenum}

Repeating this process for $n$ iterations yields interval representations and costs of the minimum-cost matchings of all sizes $1,\ldots,n$.  

{\bf Running time.}
Building the data structure $\mathcal{D}$ takes $O(pn\log^2 n)$ time. The remaining initialization, including the construction of the initial priority queue of adjacent free pairs, takes $O(n\log n)$ time. Each iteration performs a constant number of priority-queue operations and interval insertions and deletions, as well as a constant number of interval-cost queries to $\mathcal{D}$, each taking time $O(\log n)$. The $n$ iterations therefore take $O(n\log n)$ time. Thus, the construction of $\mathcal{D}$ dominates the running time. Overall, this implies \Cref{thm:algo-partial-wasserstein}.

\section{Experimental Evaluation}
\label{sec:experiments}
%\todo{J: discuss hardware? I am doing it currently on my macbook lol}
\begin{figure}
    \centering
    \includegraphics[width=0.915\linewidth]{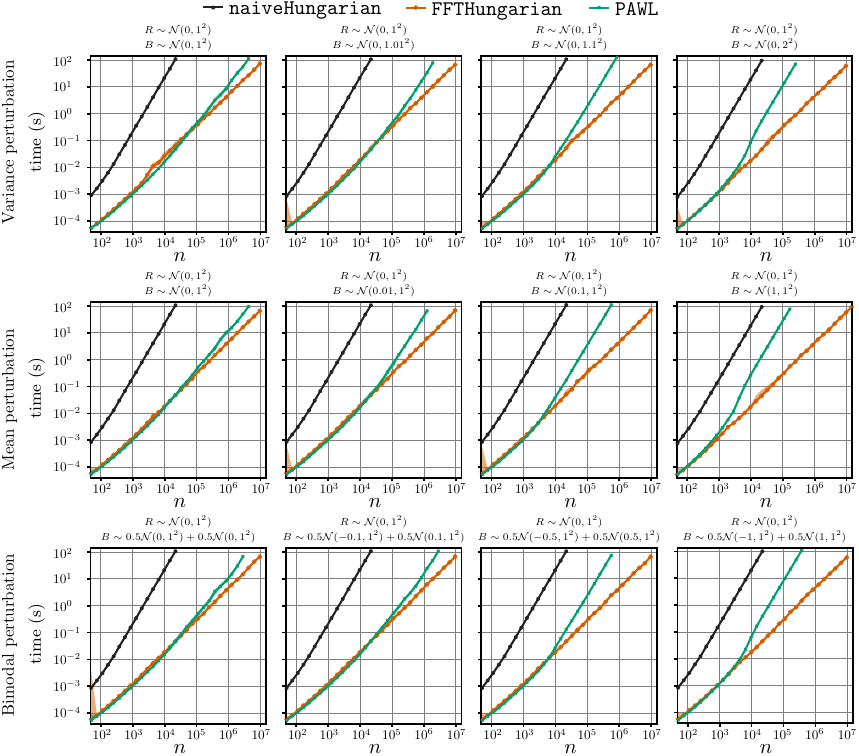}
\caption{Running-time comparison of our naive (black) and FFT-based (orange) implementations with implementation from Chapel and Tavenard~\cite{DBLP:conf/iclr/ChapelT25} (green) on synthetic data sets. Curves show median running times; shaded regions show the $10$th--$90$th percentile range across runs. Rows correspond to distribution families, with increasing distributional difference from left to right.}
    \label{fig:running_time}
\end{figure}
\begin{comment}
\begin{figure}[p]
    \centering
    \includegraphics[width=0.6\linewidth]{pictures/Running_time_same.png}
    \caption{Running time comparison of our implementation with a naive $\tilde{O}(n^2)$ implementation, and the implementation given in \cite{DBLP:conf/iclr/ChapelT25} on synthetic data sampled from normal distributions with the same mean and standard deviation. (This is a png, need to redo this...)}
    \label{fig:runningtime_same}
\end{figure}
\begin{figure}[p]
    \centering
    \includegraphics[width=0.6\linewidth]{pictures/Running_time_different_std.pdf}
    \caption{Running time comparison on synthetic data sampled from normal distributions with the same mean, but different standard deviation ($1$ vs $1.3$).}
    \label{fig:runningtime_std}
\end{figure}
\todo{J: Figure description (of the exact distributions used) is not $100\%$ correct. gonna fix later.}
\begin{figure}
    \centering
    \includegraphics[width=0.6\linewidth]{pictures/Running_time_different_mean.pdf}
    \caption{Running time comparison on synthetic data sampled from beta distributions $B(1,2)$ and $B(2,1)$, that is, same standard deviation, but different mean.}
    \label{fig:runningtime_mean}
\end{figure}
\end{comment}

To validate the efficiency and practicality of our approach, we compare our algorithm against the algorithm given in \cite{DBLP:conf/iclr/ChapelT25}, with $p=2$.

%\vspace{-1mm}
{\bf Implementation Details.} 
We consider three implementations. The first is a naive Python baseline, \texttt{naiveHungarian}, which computes the cost of an interval with $M$ points in $O(M)$ time. The second is the implementation \their of \cite{DBLP:conf/iclr/ChapelT25}, which follows the same principle but is substantially optimized using \texttt{numba}/\texttt{njit}. Our implementation \our\footnote{The code is available at \url{https://github.com/JacobusTheSecond/wasserstein}} is written in \texttt{C++} and follows the description in \Cref{sec:algorithm}, with one practical modification: for intervals containing fewer than $2048$ points, we compute the cost naively; otherwise, we query our data structure in $O(\log n)$ time. The threshold was chosen to approximately balance the two running times.

%We give a naive baseline implementation \texttt{naiveHungarian}, which is written in \texttt{python}, and computes the cost of any interval naively in $O(M)$ time, where $M$ is the number of points in the interval. The algorithm \their from \cite{DBLP:conf/iclr/ChapelT25} functions in a similar way, is however more tuned, and uses \texttt{numba}/\texttt{njit}, which significantly improves the wall-clock time. Our implementation, \our\footnote{The code is available at \url{https://gitlab.com/JacobusTheSecond/wasserstein}}, is written in \texttt{C++}, and follows the description given in \Cref{sec:algo}, except for the following modification: Whenever an interval is queried for its cost, we fall back to the naive $O(M)$ computation, whenever there are less than $1024$ points in the interval. Otherwise, we query the data structure, which takes $O(\log n)$ time. The cutoff of $1024$ was chosen to approximately balance the two running times. %Otherwise, the implementation is faithful to the theoretical description.
%\vspace{-1mm}

{\bf Experimental Setup.} We evaluate the algorithms on synthetic instances designed to cover regimes in which the input distributions range from identical to moderately different. We use samples of varying sizes from: (i) two normal distributions $\mathcal{N}(0,1^2)$ and $\mathcal{N}(0,\sigma^2)$ with different standard deviations $\sigma\in\{1,1.01,1.1,2\}$, (ii) two normal distributions $\mathcal{N}(0,1^2)$ and $\mathcal{N}(\mu,1^2)$ with different means $\mu\in\{0,0.01,0.1,1\}$, and (iii) one unimodal distribution $\mathcal{N}(0,1^2)$ and one bimodal distribution $\frac{1}{2}\mathcal{N}(-\mu,1^2)+\frac{1}{2}\mathcal{N}(\mu,1^2)$ with $\mu\in\{0,0.1,0.5,1\}$.
%(i) two normal distributions with equal mean and standard deviations $\sigma\in\{1,1.01,1.1,2\}$; (ii) two normal distributions with equal standard deviation and means $\mu\in\{0,0.01,0.1,1\}$; and (iii) a unimodal distribution $\mathcal{N}(0,1^2)$ and a bimodal distribution $\frac{1}{2}\mathcal{N}(-\mu,1^2)+\frac{1}{2}\mathcal{N}(\mu,1^2)$ with $\mu\in\{0,0.1,0.5,1\}$. 
In all cases the distributions have substantial overlap, which is favorable to the augmenting-path-based baselines. %Each instance-algorithm pair is run at least three times with a cutoff of $60$ seconds, and we report the $10$th percentile, median, and $90$th percentile running time. 
%Each instance--algorithm pair was run between $3$ and $20$ times, depending on the projected running time, with a $60$-second cutoff; we report the $10$th percentile, median, and $90$th percentile of the running times across runs.
Each instance--algorithm pair was run $3$--$20$ times, depending on its running time, with a $60$-second cutoff; we report the $10$th percentile, median, and $90$th percentile running times.
All experiments were conducted on an Apple MacBook Air with an M4 processor and $24$ GB of RAM. %; memory usage
%was never a bottleneck.
%For the experiments, we use a range of synthetic data that explores in which regime which implementation is best. In particular, we use synthetic data consisting of samples of different sizes of (i) two normal distributions $\mathcal{N}(0,1^2)$ and $\mathcal{N}(\mu,\sigma^2)$ with same mean $\mu=0$ and different standard deviation $\sigma\in\{1,1.01,1.1,2\}$, (ii) two normal distributions $\mathcal{N}(0,1^2)$ and $\mathcal{N}(\mu,\sigma^2)$ with same standard deviation $\sigma=1$ and different mean $\mu\in\{0,0.01,0.1,1\}$, and (iii) one unimodal distribution $\mathcal{N}(0,1^2)$ and one bimodal distribution $\frac{1}{2}\mathcal{N}(-\mu,1^2)+\frac{1}{2}\mathcal{N}(\mu,1^2)$ with $\mu\in\{0,0.1,0.5,1\}$. Note that in all three cases the two distributions still have significant overlap, so often times, augmenting paths are not too long. Every algorithm is ran on the same data, with a cutoff time of $60$ seconds. Each instance-algorithm combination is ran at least three times, with only the median time being reported. %\todo{J: currently waiting for this experiment to be done [ETA $6$h]... The current figure has cutoff time $1$ second} 
%All experiments were conducted on an Apple MacBook Air equipped with an M4 processor and $24$ GB of RAM. Every algorithm uses near-linear space, so the amount of RAM was at no point a bottleneck.

{\bf Evaluation.} 
The running times are shown in \Cref{fig:running_time}. We observe that \our is significantly faster than \their whenever $n\geq 10\,000$ and the two distributions are not identical. Even when the standard deviation or mean differs by only $10\%$, the running time of \their approaches its quadratic worst-case behavior, while \our continues to scale near-linearly. Across all tested regimes, \our is never substantially slower than \their and becomes increasingly faster as the input size grows. 
Overall, for large samples from even slightly different distributions, \our is the only implementation that remains practically viable. This regime is central in applications, where the goal is to compare samples from unknown distributions rather than from exactly identical ones.

\section{Lower Bounds}
\label{sec:lower-bounds}
In this section we show that the dependence on $p$ of our running time is unavoidable, by proving a quadratic lower bound for $p = \infty$. Details of our second set of lower bounds are given in \Cref{sec:omitted_lbs}.%
%%%%%%%%%%%%%%%%%%%%%%%%%%%%%%%%%%%%%%
% Removed the whole paragraph because
% it was repetitive.
%%%%%%%%%%%%%%%%%%%%%%%%%%%%%%%%%%%%%%
% In this section we present conditional lower bounds for two problems.
% First, we show quadratic hardness of computing the matching profile of the $\infty$-Wasserstein distance for points on the line under the \((\min, +)\)-Convolution Conjecture.
% Second, consider a data structures \(\mathcal D\) for the Two Intervals Matching Problem.
% We show that if \(\mathcal D\) is \enquote{combinatorial}, the product of preprocessing and query time can not be sub-quadratic, unless truly-subcubic \enquote{combinatorial} matrix multiplication is possible.
% On the other hand, even if \(\mathcal D\) uses non-combinatorial methods, it can not achieve the same bounds we achieve for Balanced Interval Matching, unless matrix multiplication is possible in near-quadratic time.
% Namely, achieving near-linear preprocessing and polylogarithmic query time is not possible under this assumption.
%\subsection{\(k\)-Partial \(\infty\)-Wasserstein Distance on the Line}
%Our lower bound is based on the conjectured hardness of \((\min,+)\)-convolution. 
\begin{definition}[$(\min,+)$-convolution]
Given two sequences $(A[i])_{i=0}^{n-1}$ and $(B[j])_{j=0}^{n-1}$ of length $n$, compute the sequence $(C[k])_{k=0}^{n-1}$ with $C[k] = \min_{i+j = k} A[i] + B[j]$.
\end{definition}%
% The \emph{\((\min,+)\)-Convolution Conjecture} states that there is no truly subquadatric algorithm for computing the \((\min, +)\)-convolution~\cite{DBLP:journals/talg/CyganMWW19}, i.e., given two integer sequences of length \(n\) with entries in \([-W, W]\), no algorithm can achieve running time \(O(n^{2-\varepsilon}\polylog W)\).
The \emph{\((\min,+)\)-Convolution Conjecture} states that, given two integer sequences of length \(n\) with entries in \([-W, W]\), no algorithm can compute their \((\min, +)\)-convolution in \(O(n^{2-\varepsilon}\polylog W)\) time for any \(\eps>0\)~\cite{DBLP:journals/talg/CyganMWW19}.
We also consider the analogously defined \((\max, -)\)-convolution. %, which simply replaces \enquote{$\min$} by \enquote{$\max$} and \enquote{$A[i] + B[j]$} by \enquote{$A[i] - B[j]$} in the definition.

\looseness=-1
In \Cref{sec:algorithm,sec:data-structure}, we show how to compute for a fixed \(p\in\N\) all \(k\)-partial \(p\)-Wasserstein distances in time \(O(n \log^{2} n)\). We now consider the case \(p = \infty\).
%We first note that our algorithm from \Cref{sec:algorithm} directly generalizes to this setting given an adapted data structure for the \(\infty\)-Wasserstein distance on a balanced interval.
%The data structure we presented in \Cref{sec:data-structure} only achieves running time of \(O(n\log^{2} n)\) preprocessing and \(O(\log n)\) query time for fixed \(p\).
%If \(p\) is part of the initial input, both terms depend linearly on~\(p\).
%The bottleneck here is that we need to compute one convolution for all integers \(0 \leq i \leq p\).
Note that our algorithm from \Cref{sec:algorithm} directly generalizes to this setting, by computing the values $W(\cdot,\cdot)$ in the data structure via \((\min, +)\)-convolutions. %, we are able to compute the \(\infty\)-Wasserstein distance of balanced intervals.
The best known \((\min, +)\)-convolution algorithm is due to Williams \cite{williams2014FasterAllpairsShortest, DBLP:journals/algorithmica/BremnerCDEHILPT14} with a running time of \(O(n^2/2^{\sqrt{\Omega(\log n)}})\), which implies a construction time of \(O(n^2/2^{\sqrt{\Omega(\log n)}})\) for our data structure, which is \emph{not} truly subquadratic.
% Note that for all \(\varepsilon > 0\), this running time is in \(\Omega(n^{2-\varepsilon})\).

We now show that this is not an artifact of our approach but rather an inherent barrier of the problem by providing a reduction from \((\min,+)\)-convolution to the problem of computing all \(k\)-partial \hbox{\(\infty\)-Wasserstein} distances.
As an intermediate step, we reduce computing the \((\min, +)\)-convolution on arbitrary instances to computing the \((\max, -)\)-convolution on structured instances.

% \todo{this lemma is given in cygan, 2.2. as a remark, the structured instances also are a trivial extension, pls check. I dont know whether we want to have it as a lemma}
% \todo{This probably depends on how familiar the target audience is with fine-grained reductions.}
\begin{lemma}\label{thm:max-minus-convolution}
  Assume the \((\min, +)\)-convolution conjecture.
  Let \((X)_{i=0}^{n-1}\) an increasing and \((Y)_{j=0}^{n-1}\) be a decreasing integer sequence with entries in \([-W, W]\).
  Then, for no \(\varepsilon>0\) there is an algorithm that takes \(X\) and \(Y\) as input and computes their \((\max, -)\)-convolution in time \(O(n^{2-\varepsilon} \polylog W)\).
\end{lemma}

% Reductions similar to \Cref{thm:max-minus-convolution} are considered standard in the literature, see, e.g., \cite{DBLP:journals/talg/CyganMWW19}.
We defer the proof to \Cref{sec:proof:max-minus-convolution}.
Finally, we reduce such structured instances of \((\max, -)\)-convolution to the computation of all \(k\)-partial \(\infty\)-Wasserstein distances, proving \Cref{thm:firstLowerbound}.
% We can now show the main result of this subsection.

%\minpluslowerbound*

\begin{proof}[Proof of \Cref{thm:firstLowerbound}]
Consider an increasing array \(X\) and a decreasing array \(Y\), each with \(n\) elements and entries from the range \([W, 2W]\) and \([0, W]\) respectively.
% This ensures that \(\min X \geq \max Y\).
Let $Z$ be their \((\max, -)\) convolution.
By \Cref{thm:max-minus-convolution}, there is no algorithm computing \(Z\) in time \(O(n^{2-\varepsilon}\polylog W)\) for any \(\varepsilon>0\).
Let \(k \in [n]\) and consider the \(k\)-partial \(\infty\)-Wasserstein distance between \(X\) and \(Y\).
Since \(\max_j Y[j] \leq \min_i X[i]\) and as \(X\) is increasing and \(Y\) is decreasing, the optimal \(k\)-matching matches the \(k\) smallest values of \(X\) to the \(k\) largest values of \(Y\).
Thus, the \(\infty\)-Wasserstein cost of the \(k\)-matching is
\(
  \max_{i+j=k}|X[i] - Y[j]|
\),
which is equal to $Z[k]$. Thus, the the vector of all \(k\)-partial \(\infty\)-Wasserstein cost of the \(k\)-matching cannot be computed in \(O(n^{2-\varepsilon}\polylog W)\) time.
%Hence, the $(\max,-)$-convolution of $X, Y$ and the vector of all \(k\)-partial \(\infty\)-Wasserstein distances of $X,Y$ are equal, 
%concluding the proof.
\end{proof}

\section*{Acknowledgments}

This research was initiated during the Workshop ``Massive Data Models and Computational Geometry~II" held at
the University of Bonn in September 2025.

Funding in direct support of this work:
\begin{compactitem}
    \item Jacobus Conradi is funded by the Carlsberg Foundation, grant CF24-1929.
    \item M\'onika Csik\'os was supported by the program “Investissement d’Avenir” launched by the French Government and implemented by ANR, with the reference ANR-18-IdEx-0001 as part of its program Emergence for the project VC-GRAPHES.
    \item Niko Hastrich received funding from the European Research
Council (ERC) under the European Unions Horizon 2020 research and innovation program (grant agreement No. 850979).
    \item Danny Mittal was supported by a travel grant from NSF 2115677, and is also partially supported by DARPA expMath, ONR MURI 2024 award on Algorithms, Learning, and Game Theory, Army-Research Laboratory (ARL) grant W911NF2410052, NSF AF:Small grants 2218678, 2114269, 2347322.
    \item André Nusser is supported by the France 2030 investment plan managed by the ANR as part of the Initiative of Excellence of Université Côte d'Azur with reference number ANR-15-IDEX-01.
    \item Work by Sharath Raghvendra is supported by an NSF Grant CCF-2514753.
\end{compactitem}

\bibliography{references.bib}

@inproceedings{rabin2011barycenter,
  author = {Julien Rabin and Gabriel Peyr{\'{e}} and Julie Delon and Marc Bernot},
  title = {Wasserstein Barycenter and Its Application to Texture Mixing},
  booktitle = {Scale Space and Variational Methods in Computer Vision - Third International Conference, {SSVM} 2011, Ein-Gedi, Israel, May 29 - June 2, 2011, Revised Selected Papers},
  year = {2011},
  publisher = {Springer},
  address = {Berlin, Heidelberg},
  pages = {435--446},
  isbn = {978-3-642-24785-9},
  timestamp = {Mon, 22 May 2017 01:00:00 +0200},
  biburl = {https://dblp.org/rec/conf/scalespace/RabinPDB11.bib},
  bibsource = {dblp computer science bibliography, https://dblp.org},
  doi = {10.1007/978-3-642-24785-9_37},
  volume = {6667},
  _bib2doi_selected = {dblp:/rec/conf/scalespace/RabinPDB11.bib},
  _bib2doi_confirmed = {true},
  _bib2doi_finished = {true},
}

@misc{nguyen2025introductionslicedoptimaltransport,
  title = {An Introduction to Sliced Optimal Transport},
  author = {Khai Nguyen},
  year = {2025},
  eprint = {2508.12519},
  archiveprefix = {arXiv},
  primaryclass = {stat.ML},
  url = {https://arxiv.org/abs/2508.12519},
  doi = {10.48550/arXiv.2508.12519},
  timestamp = {Tue, 23 Sep 2025 01:00:00 +0200},
  biburl = {https://dblp.org/rec/journals/corr/abs-2508-12519.bib},
  bibsource = {dblp computer science bibliography, https://dblp.org},
  _bib2doi_selected = {dblp:/rec/journals/corr/abs-2508-12519.bib},
  _bib2doi_confirmed = {true},
}

@inproceedings{bai2023sliced,
  title = {Sliced optimal partial transport},
  author = {Bai, Yikun and Schmitzer, Bernhard and Thorpe, Matthew and Kolouri, Soheil},
  booktitle = {Proceedings of the IEEE/CVF Conference on Computer Vision and Pattern Recognition},
  pages = {13681--13690},
  year = {2023},
  timestamp = {Tue, 29 Aug 2023 01:00:00 +0200},
  biburl = {https://dblp.org/rec/conf/cvpr/BaiSTK23.bib},
  bibsource = {dblp computer science bibliography, https://dblp.org},
  doi = {10.1109/CVPR52729.2023.01315},
  _bib2doi_selected = {dblp:/rec/conf/cvpr/BaiSTK23.bib},
  _bib2doi_confirmed = {true},
}

@book{santambrogio2015optimal,
  title = {Optimal Transport for Applied Mathematicians: Calculus of Variations, PDEs, and Modeling},
  author = {Santambrogio, Filippo},
  isbn = {9783319208282},
  series = {Progress in Nonlinear Differential Equations and Their Applications},
  url = {https://books.google.com/books?id=UOHHCgAAQBAJ},
  year = {2015},
  publisher = {Springer International Publishing},
  doi = {10.1007/978-3-319-20828-2},
  _bib2doi_finished = {true},
}

@inproceedings{DBLP:conf/iclr/ChapelT25,
  author = {Laetitia Chapel and Romain Tavenard},
  title = {One for all and all for one: Efficient computation of partial Wasserstein distances on the line},
  booktitle = {The Thirteenth International Conference on Learning Representations, {ICLR} 2025, Singapore, April 24-28, 2025},
  publisher = {OpenReview.net},
  year = {2025},
  url = {https://openreview.net/forum?id=kzEPsHbJDv},
  timestamp = {Thu, 10 Jul 2025 01:00:00 +0200},
  biburl = {https://dblp.org/rec/conf/iclr/ChapelT25.bib},
  bibsource = {dblp computer science bibliography, https://dblp.org},
  _bib2doi_selected = {dblp:/rec/conf/iclr/ChapelT25.bib},
  _bib2doi_confirmed = {true},
}

@inproceedings{williams2014FasterAllpairsShortest,
  title = {Faster All-Pairs Shortest Paths via Circuit Complexity},
  booktitle = {Proceedings of the Forty-Sixth Annual {{ACM}} Symposium on {{Theory}} of Computing},
  author = {Williams, Ryan},
  date = {2014-05-31},
  year = {2014},
  series = {{{STOC}} '14},
  pages = {664--673},
  doi = {10.1145/2591796.2591811},
  isbn = {978-1-4503-2710-7},
  timestamp = {Wed, 14 Nov 2018 00:00:00 +0100},
  biburl = {https://dblp.org/rec/conf/stoc/Williams14a.bib},
  bibsource = {dblp computer science bibliography, https://dblp.org},
  _bib2doi_selected = {dblp:/rec/conf/stoc/Williams14a.bib},
  _bib2doi_confirmed = {true},
}

@article{DBLP:journals/talg/CyganMWW19,
  author = {Marek Cygan and Marcin Mucha and Karol Wegrzycki and Michal Wlodarczyk},
  title = {On Problems Equivalent to (min, +)-Convolution},
  journal = {{ACM} Trans. Algorithms},
  volume = {15},
  number = {1},
  pages = {14:1--14:25},
  year = {2019},
  url = {https://doi.org/10.1145/3293465},
  doi = {10.1145/3293465},
  timestamp = {Tue, 16 Aug 2022 01:00:00 +0200},
  biburl = {https://dblp.org/rec/journals/talg/CyganMWW19.bib},
  bibsource = {dblp computer science bibliography, https://dblp.org},
  _bib2doi_selected = {dblp:/rec/journals/talg/CyganMWW19.bib},
  _bib2doi_confirmed = {true},
}

@inproceedings{LIPIcs.ICALP.2018.62,
  author = {Gawrychowski, Pawel and Uznanski, Przemyslaw},
  title = {{Towards Unified Approximate Pattern Matching for Hamming and L\underline1 Distance}},
  booktitle = {45th International Colloquium on Automata, Languages, and Programming (ICALP 2018)},
  pages = {62:1--62:13},
  series = {Leibniz International Proceedings in Informatics (LIPIcs)},
  isbn = {978-3-95977-076-7},
  issn = {1868-8969},
  year = {2018},
  volume = {107},
  editor = {Chatzigiannakis, Ioannis and Kaklamanis, Christos and Marx, D\'{a}niel and Sannella, Donald},
  publisher = {Schloss Dagstuhl -- Leibniz-Zentrum f{\"u}r Informatik},
  address = {Dagstuhl, Germany},
  url = {https://drops.dagstuhl.de/entities/document/10.4230/LIPIcs.ICALP.2018.62},
  urn = {urn:nbn:de:0030-drops-90669},
  doi = {10.4230/LIPIcs.ICALP.2018.62},
  annote = {Keywords: approximate pattern matching, conditional lower bounds, L\underline1 distance, Hamming distance},
  timestamp = {Thu, 14 Oct 2021 01:00:00 +0200},
  biburl = {https://dblp.org/rec/conf/icalp/GawrychowskiU18.bib},
  bibsource = {dblp computer science bibliography, https://dblp.org},
  _bib2doi_selected = {dblp:/rec/conf/icalp/GawrychowskiU18.bib},
  _bib2doi_confirmed = {true},
}

@misc{boneh2024hammingdistanceoracle,
  title = {Hamming Distance Oracle},
  author = {Itai Boneh and Dvir Fried and Shay Golan and Matan Kraus},
  year = {2024},
  eprint = {2407.05430},
  archiveprefix = {arXiv},
  primaryclass = {cs.DS},
  url = {https://arxiv.org/abs/2407.05430},
  doi = {10.48550/arXiv.2407.05430},
  timestamp = {Mon, 12 Aug 2024 01:00:00 +0200},
  biburl = {https://dblp.org/rec/journals/corr/abs-2407-05430.bib},
  bibsource = {dblp computer science bibliography, https://dblp.org},
  _bib2doi_selected = {dblp:/rec/journals/corr/abs-2407-05430.bib},
  _bib2doi_confirmed = {true},
}

@article{DBLP:journals/algorithmica/BremnerCDEHILPT14,
  author = {David Bremner and Timothy M. Chan and Erik D. Demaine and Jeff Erickson and Ferran Hurtado and John Iacono and Stefan Langerman and Mihai P{\u{a}}tra{\c{s}}cu and Perouz Taslakian},
  title = {Necklaces, Convolutions, and {X+Y}},
  journal = {Algorithmica},
  volume = {69},
  number = {2},
  pages = {294--314},
  year = {2014},
  url = {https://doi.org/10.1007/s00453-012-9734-3},
  doi = {10.1007/s00453-012-9734-3},
  timestamp = {Fri, 30 Nov 2018 00:00:00 +0100},
  biburl = {https://dblp.org/rec/journals/algorithmica/BremnerCDEHILPT14.bib},
  bibsource = {dblp computer science bibliography, https://dblp.org},
  _bib2doi_selected = {dblp:/rec/journals/algorithmica/BremnerCDEHILPT14.bib},
  _bib2doi_confirmed = {true},
}

@article{rubner2000earth,
  title = {The earth mover's distance as a metric for image retrieval},
  author = {Rubner, Yossi and Tomasi, Carlo and Guibas, Leonidas J},
  journal = {International J. of Comput. Vision},
  volume = {40},
  number = {2},
  pages = {99--121},
  year = {2000},
  publisher = {Springer Nature BV},
  timestamp = {Fri, 13 Mar 2020 00:00:00 +0100},
  biburl = {https://dblp.org/rec/journals/ijcv/RubnerTG00.bib},
  bibsource = {dblp computer science bibliography, https://dblp.org},
  doi = {10.1023/A:1026543900054},
  _bib2doi_selected = {dblp:/rec/journals/ijcv/RubnerTG00.bib},
  _bib2doi_confirmed = {true},
  _bib2doi_finished = {true},
}

@article{peyre2019computational,
  author = {Gabriel Peyr{\'{e}} and Marco Cuturi},
  title = {Computational Optimal Transport},
  journal = {Found. Trends Mach. Learn.},
  volume = {11},
  number = {5-6},
  pages = {355--607},
  year = {2019},
  url = {https://doi.org/10.1561/2200000073},
  doi = {10.1561/2200000073},
  timestamp = {Thu, 18 Jun 2020 01:00:00 +0200},
  biburl = {https://dblp.org/rec/journals/ftml/PeyreC19.bib},
  bibsource = {dblp computer science bibliography, https://dblp.org},
  _bib2doi_selected = {dblp:/rec/journals/ftml/PeyreC19.bib},
  _bib2doi_confirmed = {true},
}

@book{villani2009optimal,
  title = {Optimal transport: old and new},
  author = {Villani, C{\'e}dric},
  volume = {338},
  year = {2009},
  publisher = {Springer},
  doi = {10.1007/978-3-540-71050-9},
  _bib2doi_finished = {true},
}

@article{chapel2020partial,
  title = {Partial optimal tranport with applications on positive-unlabeled learning},
  author = {Chapel, Laetitia and Alaya, Mokhtar Z and Gasso, Gilles},
  journal = {Advances in Neural Information Processing Systems},
  volume = {33},
  pages = {2903--2913},
  year = {2020},
  timestamp = {Tue, 19 Jan 2021 00:00:00 +0100},
  biburl = {https://dblp.org/rec/conf/nips/ChapelAG20.bib},
  bibsource = {dblp computer science bibliography, https://dblp.org},
  url = {https://proceedings.neurips.cc/paper/2020/hash/1e6e25d952a0d639b676ee20d0519ee2-Abstract.html},
  _bib2doi_selected = {dblp:/rec/conf/nips/ChapelAG20.bib},
  _bib2doi_confirmed = {true},
  _bib2doi_finished = {true},
}

@inproceedings{raghvendra2024new,
  title = {A New Robust Partial p-Wasserstein-Based Metric for Comparing Distributions},
  author = {Raghvendra, Sharath and Shirzadian, Pouyan and Zhang, Kaiyi},
  booktitle = {41st Internat. Conference on Machine Learning},
  year = {2024},
  timestamp = {Mon, 09 Feb 2026 00:00:00 +0100},
  biburl = {https://dblp.org/rec/conf/icml/RaghvendraS024.bib},
  bibsource = {dblp computer science bibliography, https://dblp.org},
  url = {https://proceedings.mlr.press/v235/raghvendra24a.html},
  _bib2doi_selected = {dblp:/rec/conf/icml/RaghvendraS024.bib},
  _bib2doi_confirmed = {true},
  _bib2doi_finished = {true},
}

@inproceedings{mukherjee2021outlier,
  title = {Outlier-robust optimal transport},
  author = {Mukherjee, Debarghya and Guha, Aritra and Solomon, Justin M and Sun, Yuekai and Yurochkin, Mikhail},
  booktitle = {Internat. Conference on Machine Learning},
  pages = {7850--7860},
  year = {2021},
  organization = {PMLR},
  timestamp = {Wed, 25 Aug 2021 01:00:00 +0200},
  biburl = {https://dblp.org/rec/conf/icml/MukherjeeGSSY21.bib},
  bibsource = {dblp computer science bibliography, https://dblp.org},
  url = {http://proceedings.mlr.press/v139/mukherjee21a.html},
  _bib2doi_selected = {dblp:/rec/conf/icml/MukherjeeGSSY21.bib},
  _bib2doi_confirmed = {true},
}

@inproceedings{phatak2023computing,
  title = {Computing all Optimal Partial Transports},
  author = {Phatak, Abhijeet and Raghvendra, Sharath and Tripathy, Chittaranjan and Zhang, Kaiyi},
  booktitle = {Proc. 11th Internat. Conference on Learning Representations},
  year = {2023},
  timestamp = {Wed, 24 Jul 2024 01:00:00 +0200},
  biburl = {https://dblp.org/rec/conf/iclr/PhatakRTZ23.bib},
  bibsource = {dblp computer science bibliography, https://dblp.org},
  url = {https://openreview.net/forum?id=gwcQajoXNF},
  _bib2doi_selected = {dblp:/rec/conf/iclr/PhatakRTZ23.bib},
  _bib2doi_confirmed = {true},
  _bib2doi_finished = {true},
}

@article{km_hungarian,
  title = {The Hungarian method for the assignment problem},
  author = {Kuhn, Harold W},
  journal = {Naval research logistics quarterly},
  volume = {2},
  number = {1-2},
  pages = {83--97},
  year = {1955},
  publisher = {Wiley Online Library},
  doi = {10.1002/nav.3800020109},
  _bib2doi_finished = {true},
}

@article{vaidya1989geometry,
  title = {Geometry helps in matching},
  author = {Vaidya, Pravin M},
  journal = {SIAM Journal of Computing},
  volume = {18},
  number = {6},
  pages = {1201--1225},
  year = {1989},
  publisher = {SIAM},
  timestamp = {Sat, 27 May 2017 01:00:00 +0200},
  biburl = {https://dblp.org/rec/journals/siamcomp/Vaidya89a.bib},
  bibsource = {dblp computer science bibliography, https://dblp.org},
  doi = {10.1137/0218080},
  _bib2doi_selected = {dblp:/rec/journals/siamcomp/Vaidya89a.bib},
  _bib2doi_confirmed = {true},
}

@article{aes_sjc99,
  author = {Agarwal, Pankaj K. and Efrat, Alon and Sharir, Micha},
  title = {Vertical Decomposition of Shallow Levels in 3-Dimensional Arrangements and Its Applications},
  journal = {SIAM Journal of Computing},
  issue_date = {Dec. 1999 to Jan. 2000},
  volume = {29},
  number = {3},
  month = {dec},
  year = {1999},
  timestamp = {Sat, 30 Sep 2023 01:00:00 +0200},
  biburl = {https://dblp.org/rec/journals/siamcomp/AgarwalES99.bib},
  bibsource = {dblp computer science bibliography, https://dblp.org},
  doi = {10.1137/S0097539795295936},
  _bib2doi_selected = {dblp:/rec/journals/siamcomp/AgarwalES99.bib},
  _bib2doi_confirmed = {true},
}

@inproceedings{DBLP:conf/stoc/AbboudFKLM24,
  author = {Amir Abboud and Nick Fischer and Zander Kelley and Shachar Lovett and Raghu Meka},
  editor = {Bojan Mohar and Igor Shinkar and Ryan O'Donnell},
  title = {New Graph Decompositions and Combinatorial Boolean Matrix Multiplication Algorithms},
  booktitle = {Proceedings of the 56th Annual {ACM} Symposium on Theory of Computing, {STOC} 2024, Vancouver, BC, Canada, June 24-28, 2024},
  pages = {935--943},
  publisher = {{ACM}},
  year = {2024},
  url = {https://doi.org/10.1145/3618260.3649696},
  doi = {10.1145/3618260.3649696},
  timestamp = {Tue, 18 Jun 2024 01:00:00 +0200},
  biburl = {https://dblp.org/rec/conf/stoc/AbboudFKLM24.bib},
  bibsource = {dblp computer science bibliography, https://dblp.org},
  _bib2doi_selected = {dblp:/rec/conf/stoc/AbboudFKLM24.bib},
  _bib2doi_confirmed = {true},
}

@article{DBLP:journals/algorithmica/RodittyZ11,
  author = {Liam Roditty and Uri Zwick},
  title = {On Dynamic Shortest Paths Problems},
  journal = {Algorithmica},
  volume = {61},
  number = {2},
  pages = {389--401},
  year = {2011},
  url = {https://doi.org/10.1007/s00453-010-9401-5},
  doi = {10.1007/s00453-010-9401-5},
  timestamp = {Sun, 02 Jun 2019 01:00:00 +0200},
  biburl = {https://dblp.org/rec/journals/algorithmica/RodittyZ11.bib},
  bibsource = {dblp computer science bibliography, https://dblp.org},
  _bib2doi_selected = {dblp:/rec/journals/algorithmica/RodittyZ11.bib},
  _bib2doi_confirmed = {true},
}

@book{oppenheim1999discrete,
  title = {Discrete-time signal processing},
  author = {Oppenheim, Alan V},
  year = {1999},
  publisher = {Pearson Education India},
  _bib2doi_finished = {true},
}

@article{cooley1965algorithm,
  title = {An algorithm for the machine calculation of complex Fourier series},
  author = {Cooley, James W and Tukey, John W},
  journal = {Mathematics of computation},
  volume = {19},
  number = {90},
  pages = {297--301},
  year = {1965},
  publisher = {JSTOR},
  doi = {10.7551/mitpress/5222.003.0014},
  _bib2doi_finished = {true},
}

% \newpage
\appendix

\section{Omitted Proofs}

\subsection{Proof of \cref{lem:hungarianprop}}\label{sec:proof:hungarianprop}
\begin{proof}
The symmetric difference of two matchings is a vertex disjoint union of alternating cycles and paths. Observe that $M_k \oplus M_{k+1}$ cannot contain a cycle or non-augmenting path, otherwise we would either contradict the optimality of $M_k$ or $M_{k+1}$ or contradict the choice of $M_{k+1}$. If $M_k \oplus M_{k+1}$ contains at least $2$ augmenting paths, it must contain a path $P$ which is augmenting for $M_k$ and a path $P'$ which is augmenting for $M_{k+1}$ such that $\mathrm{cost}(M_{k+1}) = \mathrm{cost}(M_{k+1}\oplus P \oplus P')$, otherwise either $M_k$ or $M_{k+1}$ was not optimal. However $|(M_{k+1}\oplus P \oplus P') \cap M_k | > |M_{k+1} \cap M_k|$ which would contradict the choice of $M_{k+1}$.
\end{proof}

\subsection{Proof of \cref{lem:net_cost}}
\label{sec:proof_lem_net_cost}
\begin{proof}
Let $M$ be a minimum-cost $k$-matching and let $P$ be a minimum net-cost
augmenting path with respect to $M$. Let $b \in B_F$ and $r \in R_F$ be its
endpoints, and write $I_{b,r} = [\min\{b,r\}, \max\{b,r\}]$.

\medskip
First, we show that $b$ and $r$ form an adacent pair of free points.
Consider the matching $M' = M \oplus P$. By Lemma~\ref{lem:hungarianprop}, both $M$ and $M'$ are minimum-cost matchings of sizes
$k$ and $k+1$, respectively, and every edge of $P$ belongs either to $M$ or to
$M'$. By Observation~1, no edge of a minimum-cost matching can contain a free
point in its interior. If there were a free point strictly between $b$ and $r$,
then some edge of $P$ would have to cross over that free point, contradicting
Observation~1. Hence there is no free point in the interior of $I_{b,r}$, and
$(b,r)$ is an adjacent free pair by definition.

Next, we argue that all edges of $P$ are contained inside the interval $I_{b,r}$. 
Suppose otherwise, that some edge of $P$ crosses either $b$ or $r$.  By Observation~1, this edge
cannot be a matching edge of~$M$.  Without loss of generality, suppose that there is an edge that crosses $r$, and let 
$(b',r')$ with $b' < r < r'$ be the first (starting from $b$) such edge of $P$. Along $P$, the subpath from $r'$ to $r$ is an alternating path of non-negative net-cost (otherwise, alternating along this subpath alone would reduce
the cost of $M$, contradicting its optimality). Because $r$ lies strictly between $b'$ and $r'$ on the line, we have
$\|b' - r\|^p < \|b' - r'\|^p$.  Replacing the subpath from $b'$ to $r$ in $P$ by the edge $(b',r)$
thus yields another augmenting path from $b$ to $r$ with net-cost strictly smaller than that of~$P$, a contradiction.

% Next, we argue that all edges of $P$ are contained inside the interval $I_{b,r}$.
% Suppose, for contradiction, that no minimum net-cost augmenting path lies
% entirely within $I_{b,r}$.  
% Among all minimum net-cost augmenting paths, choose $P$ to minimize the number
% of crossings of the boundary of $I_{b,r}$.  
% Then some edge of $P$ must cross either $b$ or $r$.  By Observation~1, this edge
% cannot be a matching edge of~$M$, so it must be a non-matching edge.  
% Without loss of generality, assume the edge crosses $r$, so it is of the form
% $(b',r')$ with $b' < r < r'$.

% Along $P$, the subpath from $r'$ to $r$ is an alternating path.  Its net-cost
% cannot be negative; otherwise, alternating along this subpath alone would reduce
% the cost of $M$, contradicting its optimality.  
% Therefore, the net-cost of the subpath of $P$ starting from $b$ to $b'$ followed by the edge
% $(b',r')$ is at most the minimum net-cost value.

% Because $r$ lies strictly between $b'$ and $r'$ on the line, we have
% $\|b' - r\|^p < \|b' - r'\|^p$.  Replacing the edge $(b',r')$ in $P$ by $(b',r)$
% thus yields another augmenting path from $b$ to $r$ with net-cost no larger than
% that of~$P$, but with strictly fewer crossings of the interval boundary.  This
% contradicts the choice of $P$.  
% Hence every edge of $P$ lies within $I_{b,r}$.

Since $P$ is contained in $I_{b,r}$, the only edges of $M$ affected by the
augmentation lie inside this interval.  By the interval representation
(Property~\ref{property:free-pts}), at most one interval $I \in \mathcal{I}_M$ is contained in
$I_{b,r}$.
If no such interval exists, then $I_{b,r}$ contains only the two free endpoints
$b$ and $r$, so $P$ consists of the single edge $(b,r)$ and
$\Phi(P)=\|b-r\|^p$.
If such an interval $I$ exists, then inside $I$ the matching $M$ coincides with
the optimal monotone matching of cost $w^p(I)$ (Property~\ref{property:monotone}), while after
augmenting along $P$ the matching inside $I_{b,r}$ becomes the optimal monotone
matching on that larger interval, with cost $w^p(I_{b,r})$.  Thus
$\Phi(P)=w^p(I_{b,r}) - w^p(I)$.
\end{proof}

\subsection{Proof of \cref{thm:max-minus-convolution}}\label{sec:proof:max-minus-convolution}
\begin{proof}
  We provide a linear time reduction from the \((\min, +)\)-convolution problem.
  Let \(A\) and \(B\) be \(n\)-element integer sequences with elements in the range \([-M, M]\) and denote their \((\min, +)\)-convolution with \(C\).
  Now set
  \[
    X[i] \coloneqq 3iM - A[i]\quad \text{and}\quad Y[j] \coloneqq -3(j + 1)M + B[j].
  \]
  Note that \(X[i] \in [3iM - M, 3iM + M]\).
  Since \(3iM + M < 3(i + 1)M - M\), \(X\) is increasing.
  Similarly, \(Y\) is decreasing.
  Let \(Z\) be the \((\max, -)\)-convolution of \(X\) and \(Y\).
  For all \(k \in \{0, \dots, n-1\}\), we have
  \begin{align*}
    Z[k]&=\max_{i+j=k} X[i] - Y[j]\\
    &=\max_{i+j=k} 3iM - A[i] -(-3(j+1)M + B[j])\\
    &=3(k+1)M + \max_{i+j=k} (- A[i] - B[j])\\
    &=3(k+1)M - \min_{i+j=k} (A[i] + B[j])\\
    &=3(k+1)M - C[k].
  \end{align*}
  Therefore, given \(Z\), we can recover \(C\) in linear time.
  Notice that all elements of \(X\) and \(Y\) are contained in \([-6Mn, 6Mn]\).
  So, if there was an \(\varepsilon>0\) such that we can compute \(Z\) in time \(O(n^{2-\varepsilon}\polylog (6Mn))\), we directly get an \(O(n^{2-\varepsilon/2}\polylog M)\) algorithm for computing \(C\), which is ruled out by the \((\min, +)\)-convolution conjecture.
\end{proof}

\section{Fast Fourier Transform}\label{appendix:fft}

In this section, we elaborate on the specific way, we apply the fast Fourier transform to compute $W(\mathrm{left}_S,\mathrm{right}_S,d)$.

Let $S = [a,b]$ be a node of the range tree. Suppose the red points in
$\mathrm{left}_S$ are $r_u,\dots,r_v$ in the global red order, and the
blue points in $\mathrm{right}_S$ are $b_s,\dots,b_t$ in the global blue
order. For any shift $d$, the value $W(\mathrm{left}_S,\mathrm{right}_S,d)$
is the sum of $|b_{j+d} - r_j|^p$ over all $j$ with
$j \in [u,v]$ and $j+d \in [s,t]$. Since all red points in
$\mathrm{left}_S$ lie to the left of all blue points in
$\mathrm{right}_S$, we always have $b_{j+d} > r_j$, so
$|b_{j+d} - r_j|^p = (b_{j+d} - r_j)^p$.

We compute these values simultaneously for all relevant shifts $d$ by
expanding each term using the binomial identity
\[
    (b_{j+d} - r_j)^p
    = \sum_{c=0}^p \binom{p}{c} (-1)^c r_j^c\, b_{j+d}^{\,p-c}.
\]
Thus, to obtain $W(\mathrm{left}_S,\mathrm{right}_S,d)$, it suffices to
compute, for each $c \in \{0,\dots,p\}$, the sum of $r_j^c\, b_k^{\,p-c}$
over all pairs $(j,k)$ with $j \in [u,v]$, $k \in [s,t]$, and
$k - j = d$.

To accomplish this efficiently, we construct polynomials encoding the
sequences of red and blue values. For each $c$, define
\[
    P_c(x)
    = r_u^c x^{\,v-u}
      + r_{u+1}^c x^{\,v-u-1}
      + \cdots
      + r_v^c,
\]
so that the $(i$th$)$ red point from the left contributes to a monomial
determined by its offset from $v$. Similarly define
\[
    Q_c(x)
    = b_s^{\,p-c}
      + b_{s+1}^{\,p-c} x
      + \cdots
      + b_t^{\,p-c} x^{\,t-s},
\]
so that the $(i$th$)$ blue point from the left contributes to
degree~$i$.
The product $P_c(x)Q_c(x)$ can be computed in
$O((b-a)\log(b-a))$ time using FFT. The coefficient of $x^\ell$ in this
product equals
\[
    \sum_{\substack{j \in [u,v],\, k \in [s,t]\\
                    (v-j) + (k-s) = \ell}} r_j^c \, b_k^{\,p-c}.
\]
The condition $(v-j) + (k-s) = \ell$ is equivalent to
$k - j = \ell + s - v$. Thus the coefficient of $x^\ell$ aggregates
exactly the contributions of all pairs $(r_j,b_k)$ in
$\mathrm{left}_S \times \mathrm{right}_S$ whose shift is
$d = \ell + s - v$.
Summing over all $c$ with the appropriate binomial coefficients, we
define
\[
    H(x) = \sum_{c=0}^p \binom{p}{c} (-1)^c\, P_c(x)Q_c(x).
\]
By construction, the coefficient of $x^\ell$ in $H(x)$ is precisely
$W(\mathrm{left}_S,\mathrm{right}_S,d)$ for $d = \ell + s - v$. Thus, to
store all cross-term values for node $S$, we store the coefficients of
$H(x)$ indexed by their corresponding shifts~$d$.

The total time to compute $W(\mathrm{left}_S,\mathrm{right}_S,d)$ for all
relevant $d$ is $O((b-a)\log(b-a))$.

\section{Omitted Lower Bounds}
\label{sec:omitted_lbs}

\subsection{Two Intervals Matching Problem}

We show two more lower bounds.
Although the data structure in our balanced-interval setting admits a near-linear-time solution for any fixed $p$, a slightly more general version of the data structure might not exist with similar performance bounds. 
In this more general formulation, the data structure is given $p \in \mathbb{N}$ and point sets $B, R \subset \mathbb{R}$, and it must answer queries of the following form: 
given two intervals $I_B, I_R \subset \mathbb{R}$ with $|B \cap I_B| = |R \cap I_R|$, return the $p$-Wasserstein distance between $B \cap I_B$ and $R \cap I_R$. 
We refer to this as the \emph{Two Intervals Matching Problem}. 

If only combinatorial\footnote{Roughly speaking, these are algorithms that avoid algebraic techniques such as fast matrix multiplication or FFT-based methods.} methods are allowed, we present a lower bound based on the Combinatorial Matrix Multiplication Conjecture. This conjecture states that no truly subcubic \emph{combinatorial} matrix multiplication algorithm exists.
\begin{restatable}[]{theorem}{hamminglowerbound}\label{thm:hamminglowerbound}
    Unless the Combinatorial Matrix Multiplication Conjecture fails, for no \(p > 1\) and \(\varepsilon > 0\) there is a combinatorial data structure for the Two Intervals Matching Problem with preprocessing time $P$ and query time $Q$ such that $P \cdot Q \in O((|B| \cdot |R|)^{1-\eps})$.
\end{restatable}
When non-combinatorial techniques are permitted, we can still rule out almost-linear preprocessing time with $n^{o(1)}$ query time unless matrix multiplication can be performed in almost-quadratic time (i.e., $\omega = 2$).
More precisely, we obtain the following lower bound:
\begin{restatable}[]{theorem}{hamminglowerboundNonCombinatorial}\label{thm:hamminglowerboundNonCombinatorial}
    Let \(n \coloneqq |B| = |R|\).
    Any data structure for the Two Intervals Matching Problem with preprocessing time $P$ and query time $Q$ satisfies $P + nQ \in \Omega(n^{\omega/2})$, where \(\omega\) denotes the Matrix Multiplication Exponent.
\end{restatable}

%We now show lower bounds for data structures solving the Two Intervals Matching Problem, first focusing on the case that this data structure is combinatorial.
Our lower bound is based on the following conjecture that has been stated many times (\newstuff{see, e.g.,~\cite{LIPIcs.ICALP.2018.62,DBLP:journals/algorithmica/RodittyZ11, DBLP:conf/stoc/AbboudFKLM24})}:
% Our lower bound is based on the following conjecture that has been stated many times (\newstuff{see, e.g.,~\cite{LIPIcs.ICALP.2018.62,DBLP:conf/soda/BackursIS17,DBLP:conf/icalp/KunnemannPS17,DBLP:journals/dcg/ChanH21, DBLP:journals/algorithmica/RodittyZ11, DBLP:conf/stoc/AbboudFKLM24,DBLP:conf/icalp/BringmannC22})}:

\begin{conjecture}[Combinatorial Matrix Multiplication]
For any $\alpha, \beta, \gamma, \eps > 0$, there is no
combinatorial algorithm for multiplying an $n^\alpha \times n^\beta$ matrix with an $n^\beta \times n^\gamma$ matrix in time
$O(n^{\alpha+\beta+\gamma-\eps})$.
\end{conjecture}
\newstuff{
As has been noted in~\cite{LIPIcs.ICALP.2018.62}, there is no \emph{precise} definition of what combinatorial means.
However, the algebraic tricks that are used in FFT as well as any currently known fast matrix multiplication algorithm are clearly non-combinatorial by the above conjecture.
}

Combinatorial Matrix Multiplication reduces to the Hamming Distance Oracle Problem~\cite{boneh2024hammingdistanceoracle}, which is defined as follows.
Given two binary strings $S, T$, preprocess these strings to allow for queries that ask, given two length $\ell$ substrings of $S$ and $T$, what is their Hamming distance.%\todo{do we mind that [2] only is an arxiv-preprint since 2024? Have we checked it manually?}
\begin{theorem}[\newstuff{by {\cite[Theorem~3]{boneh2024hammingdistanceoracle}}}]
Unless the Combinatorial Matrix Multiplication Conjecture fails, there is no combinatorial data structure for the Hamming Distance Oracle Problem with query time $Q$ and preprocessing time $P$ such that $P \cdot Q \in O((|S| \cdot |T|)^{1-\eps})$.
\end{theorem}

Recall the \emph{Two Intervals Matching Problem}.
The data structure is given \(p \in \N\) and two point sets \(B, R \subset \R\) to preprocess. Now it should answer the following queries: Given two contiguous intervals $I_B, I_R \subset \R$ such that $|B \cap I_B| = |R \cap I_R|$, return the smallest matching cost between $B \cap I_B$ and $R \cap I_R$.
We obtain the following lower bound:
\hamminglowerbound*

\begin{proof}
% Define 0 and 1 gadgets
In the following we use the notation $S + x \coloneqq \{s + x \mid s \in S\}$, where $S \subset \R$ and $x \in \R$.
We first define the gadgets that we use to encode 0s and 1s:
\[
G(0) = \{1, 4\},\quad G(1) = \{2, 3\}.
\]
To understand why we choose these gadgets, consider their $p$-Wasserstein distances when placed some distance $M$ apart:
\begin{align*}
    Y(M) &\coloneqq w(G(0), G(0) + M) = w(G(1), G(1) + M) = 2M^p\\
    Z(M) &\coloneqq w(G(0), G(1) + M) = w(G(1), G(0) + M) = (M+1)^p + (M-1)^p > 2M^p.
\end{align*}
Note that gadgets of the same type result in a smaller distance than matching gadgets of a different type. Hence, this can encode the Hamming Distance on bit level.

% Define A and B
Let $S \in \{0,1\}^n$ and $T \in \{0,1\}^m$ be the two input strings of the Hamming Distance Oracle Problem.
Now we define the point sets $B, R \subset \N$ from the Two Intervals Matching Problem that we use to encode the input of the Hamming Distance Oracle Problem.
\[
B \coloneqq \{ G(S[i]) + 4i \mid i \in \{1, \dots, n\} \},\quad R \coloneqq \{ G(T[i]) + 4i \mid i \in \{1, \dots, m\} \} + 4n.
\]
% This finishes the reduction from $S,T$ to $A,B$.
% Define how queries translate
It remains to specify how we reduce the queries.
A query for the Hamming Distance between $S[i..i+\ell-1]$ and $T[j..j+\ell-1]$ is transformed into a query to the Two Intervals Matching Problem with intervals $I_B = [4i+1,4(i+\ell)]$ and $I_R = [4j+1,4(j+\ell)]$.
% This returns a distance of the form $(4n)^2 + 2*H$, where $H$ is the desired Hamming Distance for the query to $S$ and $T$.
As the optimal matching matches the points in order, this query returns a distance of the form $k \cdot Z(4|j - i|) + (\ell-k) \cdot Y(4|j - i|)$, where $k$ is the Hamming Distance between $S[i \dots i+\ell]$ and $T[j \dots j+\ell]$.

This finishes the description of the reduction.
As $|B| \in \Theta(|S|)$ and $|R| \in \Theta(|T|)$, the lower bound for the Hamming Distance Oracle Problem simply carries over.
\end{proof}

Now we turn our attention to lower bounds to data structures that may utilize non-combinatorial approaches.
We notice that the reduction used by \cite{boneh2024hammingdistanceoracle} to obtain the combinatorial lower bound for Hamming Distance Oracle can actually be directly used to obtain a non-combinatorial lower bound.
Concretely, given two binary $n \times n$ matrices $A,B$ whose product $C$ we want to compute, we can encode each of \(A\) and \(B\) in a binary string of length $O(n^2)$ using the same encoding as in~\cite{boneh2024hammingdistanceoracle}.
Using a single Hamming distance query to $S,T$, we can thereby compute one entry of $C$; hence, we can compute all of $C$ using $n^2$ Hamming distance queries.
Let $N \coloneqq n^2$.
The above discussion implies that for the Hamming Distance Oracle problem on two strings of length $O(N)$, the preprocessing time plus $N$ queries cannot be faster than $\Omega(n^\omega) = \Omega(N^{\omega/2})$.
Thus, with the same reduction that we applied to obtain \Cref{thm:hamminglowerbound}, we get our non-combinatorial lower bound.

\hamminglowerboundNonCombinatorial*

At this point, we want to emphasize that for the Balanced Interval Matching Problem, we achieved near-linear preprocessing- and logarithmic query time.
\Cref{thm:hamminglowerboundNonCombinatorial} shows that this and even almost-linear preprocessing- and subpolynomial query time is not possible, unless matrix multiplication is possible in almost-quadratic time.

\end{document}